\documentclass[11pt]{article}

\usepackage{amsmath, amssymb, amsfonts, amsthm}
\usepackage[english]{babel}
\usepackage{bm}
\usepackage{mathrsfs}
\usepackage{enumerate}
\usepackage{dsfont}
\usepackage{color}

\numberwithin{equation}{section}

\newtheorem{theorem}{Theorem}[section]
\newtheorem{proposition}[theorem]{Proposition}

\newtheorem{lemma}[theorem]{Lemma}

\newtheorem{remark}[theorem]{Remark}

\newcommand{\bdm}{\begin{displaymath}}
\newcommand{\edm}{\end{displaymath}}
\newcommand{\bdn}{\begin{eqnarray}}
\newcommand{\edn}{\end{eqnarray}}
\newcommand{\bay}{\begin{array}{c}}
\newcommand{\eay}{\end{array}}
\newcommand{\ben}{\begin{enumerate}}
\newcommand{\een}{\end{enumerate}}
\newcommand{\beq}{\begin{equation}}
\newcommand{\eeq}{\end{equation}}
\newcommand{\bml}[1]{\begin{multline} #1 \end{multline}}
\newcommand{\bmln}[1]{\begin{multline*} #1 \end{multline*}}

\newcommand{\one}{\mathds{1}}

\newcommand{\hilb}{\mathscr{H}}
\newcommand{\f}{\frac}
\newcommand{\F}{\mathcal{F}}
\newcommand{\pot}{\mathcal{G}}

\newcommand{\formo}{\F_{0,\betav}}

\newcommand{\formol}{\F_{0,\beta}^{(n)}}
\newcommand{\formfo}{\F_{0}}

\newcommand{\hamo}{H_{0,\betav}}
\newcommand{\hamol}{H_{0,\beta}^{(n)}}
\newcommand{\hamfree}{H_{\mathrm{free}}}

\newcommand{\uform}{F}

\newcommand{\qformla}{\Phi_{\la}}
\newcommand{\qformo}{\Phi_{0}}

\newcommand{\gammaol}{\Gamma_0}

\newcommand{\xiplus}{\tilde{\xi}^{+}}

\newcommand{\ximinus}{\tilde{\xi}^{-}}

\newcommand{\xximinus}{\Xi^{-}_{n}}

\newcommand{\dom}{\mathscr{D}}

\newcommand{\N}{\mathbb{N}}
\newcommand{\R}{\mathbb{R}}

\newcommand{\ldf}{L^{2}_{\mathrm{f}}}
\newcommand{\hf}{H^{1}_{\mathrm{f}}}
\newcommand{\hdf}{H^{2}_{\mathrm{f}}}

\newcommand{\lf}{\left}
\newcommand{\ri}{\right}

\newcommand{\xv}{\mathbf{x}}

\newcommand{\yv}{\mathbf{y}}

\newcommand{\qv}{\mathbf{q}}

\newcommand{\kv}{\mathbf{k}}

\newcommand{\kvp}{\mathbf{k}^{\prime}}

\newcommand{\pv}{\mathbf{p}}

\newcommand{\betav}{\bm{\beta}}

\newcommand{\diff}{\mathrm{d}}
\newcommand{\la}{\lambda}

\newcommand{\mstar}{m^{\star}}

\newcommand{\mstaroo}{m^{\star}}
\newcommand{\mss}{m^{\star\star}}

\newcommand{\mssol}{m^{\star\star}_{\ell}}
\newcommand{\mssoo}{m^{\star\star}}

\newcommand{\betal}{\beta_{n}}

\newcommand{\tx}{\textstyle}
\newcommand{\disp}{\displaystyle}

\newtheorem{teo}{Theorem}[section]
\newtheorem{lem}{Lemma}[section]
\newtheorem{pro}{Proposition}[section]

\newtheorem{defi}{Definition}[section]

\newcounter{remark}[section]

\newcommand{\be}{\begin{equation}}
\newcommand{\ee}{\end{equation}}
\newcommand{\bey}{\begin{eqnarray}}
\newcommand{\eey}{\end{eqnarray}}

\newcommand{\eps}{\varepsilon}

\newcommand{\bC}{{\mathbb C}}

\input epsf

\newcommand{\donothing}[1]{}
\newcommand{\n}{\noindent}
\newcommand{\vs}{\vspace{0.5 cm}}

\usepackage{bbm}

\begin{document}





\title{A Class of Hamiltonians for a Three-Particle Fermionic System at Unitarity}

\author{M. Correggi$^{1}$, G. Dell'Antonio$^{2,3}$, D. Finco$^{4}$, A. Michelangeli$^{3,5}$, A. Teta$^2$   \\ \\
\normalsize{{\it 1. Dipartimento di Matematica e Fisica, Universit\`a degli Studi Roma Tre}}\\
\normalsize{{\it Largo San Leonardo Murialdo 1, 00146 Roma, Italy}}\\
\normalsize{{\it 2. Dipartimento di Matematica, ''Sapienza'' Universit\`a di Roma}}\\
\normalsize{{\it P.le A. Moro 5, 00185 Roma, Italy}}\\
\normalsize{{\it 3. Scuola Internazionale Superiore di Studi Avanzati}}\\
\normalsize{{\it Via Bonomea  265, 34136 Trieste, Italy}}\\
\normalsize{{\it 4. Facolt\`a di Ingegneria, Universit\`a Telematica Internazionale Uninettuno}}\\
\normalsize{{\it Corso V. Emanuele II 39,  00186 Roma, Italy}}\\
\normalsize{{\it 5. Institute of Mathematics, LMU Munich}}\\
\normalsize{{\it Theresienstr. 39, 80333 Munich, Germany}}}

\date{\today}

\maketitle

\begin{abstract}
We consider a quantum mechanical three-particle system made of two identical fermions of mass one and a different particle of mass $m$, where each fermion interacts via a zero-range force with the different particle. In particular we study the unitary regime, i.e., the case of infinite two-body scattering length. The   Hamiltonians describing the system are, by definition,  self-adjoint extensions  of the free Hamiltonian restricted on smooth functions vanishing at the two-body coincidence planes, i.e., where the positions of two interacting particles coincide.

\n
It is known that  for $m$ larger than a critical value  $m^* \simeq (13.607)^{-1}$  a self-adjoint and lower bounded  Hamiltonian $H_0$ can be constructed, whose domain is characterized in terms of the standard point-interaction boundary condition at each coincidence plane. 

\n
Here we prove that for $m\in(m^*,m^{**})$, where $m^{**}\simeq  (8.62)^{-1}$, there is a further family of self-adjoint and lower bounded Hamiltonians $H_{0,\beta}$, $\beta \in \R$, describing the system. Using a quadratic form method, we give a rigorous construction of such Hamiltonians and we show that the elements of their domains satisfy a further  boundary condition, characterizing the singular behavior when the positions of all the three particles coincide.
\end{abstract}

\section{Introduction}

The theoretical analysis of the quantum mechanical three-body problem with pairwise zero-range interactions is a subject to considerable interest in the physics of cold atoms. This is essentially due to the recently achieved  possibility to realize experimental conditions where the interaction is well described by a zero-range force, in particular in the unitary limit. Roughly speaking,  unitary limit means  that the two-body interaction is characterized by a zero-energy resonance or, equivalently, by an infinite value of the scattering length.    The correct definition of the model, the occurrence of the Efimov effect and the analysis of the stability problem, i.e., the existence of a finite lower bound for the Hamiltonian, have been widely studied both in the physical \cite{bh,cmp,ct,cw,km,tc,wc1,wc2} and in the mathematical \cite{CDFMT,DFT,fm,FT,ms,m1,m3,m4} literature. 

Here we consider the $2+1$ fermionic  problem, where two identical fermions of mass one interact with a particle of different nature and mass $m$ at unitarity. Setting $\hbar=1$, the formal Hamiltonian of the system reads
\beq
	\mathcal{H}=- \frac{1}{2m} \Delta_{\xv_0} - \frac{1}{2} \Delta_{\xv_1} - \frac{1}{2} \Delta_{\xv_2} + \mu \delta(\xv_1 - \xv_0) + \mu \delta(\xv_2 - \xv_0),
\eeq
where the  fermions are labelled by $ 1, 2$ and $\mu \in \R$. We denote vectors in $ \R^d $  by bold-face symbols $ \xv $, while we will set $ x = |\xv| $. Extracting the center of mass motion and introducing the relative coordinates $\yv_i=\xv_0 - \xv_i$, $i=1,2$,  we can reduce to study the formal operator
\beq\label{formal ham}
H= - \Delta_{\yv_1} - \Delta_{\yv_2} - \f{2}{m+1} \nabla_{\yv_1} \cdot \nabla_{\yv_2} +\mu \delta (\yv_1) + \mu \delta (\yv_2).
\eeq
Due to the symmetry constraint, the Hamiltonian acts on the Hilbert space given by square integrable functions which are antisymmetric under exchange $ \yv_1 \to \yv_2 $, i.e.,
\beq
	\label{Hilbert}
	\ldf(\R^6) = \lf\{ \lf. \Psi \in L^2(\R^6) \: \ri| \: \Psi(\yv_1,\yv_2) = - \Psi(\yv_2,\yv_1) \ri\}.
\eeq


There are several possible ways to give  the formal expression \eqref{formal ham} a mathematically rigorous meaning, where the formal coupling constant $\mu$ must be replaced by a new renormalized parameter. 
A typical approach which exploits the theory of self-adjoint extensions of symmetric operators goes through the analysis of the symmetric operator given by the free kinetic energy of the three particles acting on the domain of regular functions vanishing on the planes $ \{ \yv_i = 0 \} $, $ i = 1,2 $, i.e., where the point interaction is supported. Such operators are not self-adjoint but only symmetric, and they admit a huge number of self-adjoint extensions (both deficiency indices equal $ + \infty $).  Among all possible extensions a special role is played by a rather small class, the so-called Skornyakov Ter-Martirosyan (STM) operators, which are the natural generalization to the three-body case of the Schr\"odinger  operator with a point interaction in the two-body case, and come up with a domain of functions that, as we shall discuss in a moment, have the expected asymptotic behaviour prescribed by physical heuristics, whenever two particles come on top of each other.

This is a major point that we want to emphasize already at this qualitative stage, before proceeding with the details and remarking it further once the appropriate notation will be set up. Indeed, the result of our work here is two-fold. On the one hand we construct a class of self-adjoint operators for our three-body system (see Proposition \ref{pro: STM}), by means of the corresponding quadratic forms, which are all of the STM form, namely reproduce the ``physical'' boundary condition in the vicinity of the coincidence planes $ \{ \yv_i=0\} $, and in which the scattering length of the two-body, zero-range interaction is set to infinity (the so-called ``unitary regime''). On the other hand, we show that, precisely in the same regime of masses determined in the physical literature through formal arguments, in the domain of each such Hamiltonian certain (``most singular'') wave-functions display a further asymptotic behaviour in the vicinity of the \emph{triple coincidence point} $ \{ \yv_1 = \yv_2 = 0 \} $, a behaviour that we can cast in the form under which is usually known in the physical literature (see Proposition \ref{pro: charge asympt} and Remark \ref{rem: triple point}).

The existence, for a special regime of masses, of additional STM extensions besides the natural (``Friedrichs'') extension that we constructed and studied in a previous work of ours \cite{CDFMT} was already known in the case of finite scattering length. We now build these extra STM extensions also at unitarity, we reproduce for each of them the physical triple-point asymptotics, and we show that in the present case of infinite scattering length the regime of masses for the existence of such STM extensions is \emph{larger} than the corresponding regime found recently in the mathematical literature \cite{m3,m4} for finite scattering length, and it coincides precisely with the regime of masses predicted by physical heuristics.

In order to develop these arguments, let us quickly revisit first the construction of a point interaction between two particles in 3 dimensions and then the natural STM generalisation for three particles. For a two-particle system   it is known (see, e.g., \cite{al}) that,  extracting the center of mass motion  and denoting by $\xv$ the relative coordinate,    the self-adjoint operator describing the Hamiltonian $h_{\alpha}$ with zero-range interaction has a domain consisting of functions $ \Psi(\xv) $, which have the following asymptotics when $ \xv \to 0$:
\beq\label{bcde}
	 \Psi(\xv) =  \frac{q}{4 \pi |\xv|} + \alpha q + o(1),	\qquad	\mbox{as } |\xv| \to 0,
\eeq
where $ q \in \bC $ is a complex number uniquely associated with $ \Psi $ and $ \alpha \in \R $ labels the self-adjoint extension. Moreover, $h_{\alpha}$ acts as the free Laplacian outside the origin. More precisely, the Hamiltonian can be defined as follows
\begin{eqnarray}
	\label{one center delta}
	\dom(h_{\alpha}) & = & \lf\{ \lf. \psi \in L^2(\R^3) \: \ri| \: \psi = \phi + q G, \, \phi \in \dot{H}^2(\R^3), \, q \in \bC, \, \phi(0) =  \alpha q \ri\},	\nonumber	\\
	h_{\alpha} \psi & = & - \Delta \phi,
\end{eqnarray}
where
\beq\label{fgre}
	G(\xv) : = \frac{1}{4 \pi |\xv|}
\eeq
and $ \dot{H}^n(\R^d) $ denotes the homogeneous space of functions $ u $ such that $ |\kv|^n \hat{u}(\kv) \in L^2(\R^d) $. We recall that functions in  $ \dot{H}^n(\R^d) $ are in particular continuous, so that the value of the function at the origin $ \phi(0) $ is well defined. Moreover it can be easily seen that the  condition  $\phi(0)=\alpha q$ in (\ref{one center delta}) is equivalent to (\ref{bcde}). 
The physical meaning of the parameter $\alpha$ is related to the notion of scattering length, which for such model  equals $(-4\pi \alpha)^{-1}$. The free Hamiltonian is recovered for $\alpha \rightarrow \infty$, while particularly relevant for our purposes is the case $\alpha =0$, corresponding to infinite scattering length. Notice that the Hamiltonian $ h_0 $  admits a zero-energy resonance,  provided by the function (\ref{fgre}).  Indeed $ G $ locally belongs to $ \dom(h_0) $ or, more precisely,  it satisfies all the conditions in \eqref{one center delta} but the required decay at $ |\xv| \to \infty $ to ensure that the function belongs to $ L^2(\R^3) $ ($ G \notin L^2(\R^3) $ but $ G \in L^2_{\mathrm{loc}}(\R^3) $). Moreover, according to the action of the operator described in \eqref{one center delta}
\bdm
	h_0 G = 0.
\edm

\n
The above considerations lead to define a STM operator $ \tilde{H}_{\alpha} $ for our $2+1$ fermionic system on a domain given by functions $\Psi$ belonging to  
\bdm
	 \hdf \lf(\R^6 \setminus \cup_{i=1,2} \{ \yv_i =0\} \ri) \cap \ldf(\R^6),
\edm
such that
\beq
	\label{plane asympt}
	\Psi(\yv_1,\yv_2) =  \frac{(-1)^{i+1}\xi(\yv_j)}{4 \pi | \yv_i|} + \alpha (-1)^{i+1}\xi(\yv_j) + o(1),	\qquad	\mbox{as } | \yv_i| \to 0, \: i,j = 1,2, \, i \neq j,
\eeq
where $\xi$ is a smooth (i.e., $C_0^{\infty}(\R^3)$) complex function defined on $\R^3$, uniquely associated with $\Psi$. Moreover, $\tilde{H}_{\alpha}$ acts as the free Hamiltonian outside the planes $\{\yv_i=0\}$. 
The STM operator is symmetric but not self-adjoint and then one can investigate the existence of  possible self-adjoint extensions. 

In \cite{CDFMT} we have approached the problem using  the theory  of quadratic forms in Hilbert spaces (for a different approach see, e.g., \cite{m1,m3,m4}). In particular we have introduced the following quadratic form which can be easily seen (see, e.g., \cite{cft}) to be the most natural one associated to the STM operator $ \tilde{H}_{0} $
\begin{eqnarray}
	\label{formfo domain}
 	\dom[\formfo] &=& \lf\{ \lf. u \in \ldf(\R^6) \: \ri| \: u = w + \pot \xi, \, w \in \dot{\hf}(\R^6), \,\xi \in \dot{H}^{1/2}(\R^3) \ri\},	\\
 	\formfo[u] & = & \uform[w] + 2 \qformo[\xi],	\label{forac}		\\
 	\uform[w] & = & \int_{\R^6} \diff \kv_1 \diff \kv_2 \: \lf( k_1^2 + k_2^2 + \tx\frac{2}{m+1} \kv_1 \cdot \kv_2 \ri) \lf| \hat{w}(\kv_1,\kv_2) \ri|^2,	\label{uform} \\
 	\qformo[\xi] &= & 2\pi^2 \frac{\sqrt{m(m+2)}}{m+1} \int_{\R^3} \diff \pv \: p \: |\hat{\xi}(\pv)|^2 + \int_{\R^6} \diff \pv \diff \qv  \: \frac{\hat{\xi}^*(\pv) \hat{\xi}(\qv)}{p^2 + q^2 + \frac{2}{m+1} \pv \cdot \qv}, \label{qformo}
\end{eqnarray}
where $\hat{f}$ denotes the Fourier transform of $f$ (recall  that the label $ 0 $ does not stand for the free Hamiltonian but rather for the Hamiltonian with $\alpha=0$, i.e., with  infinite two-body scattering length). Moreover the ``potential''  generated by the ``charge'' $ \xi $ (we will often use this terminology borrowed from electrostatics) is  defined by
\beq
	\label{potential}
	\lf(\widehat{\pot \xi}\ri)\lf(\kv_1,\kv_2\ri) = \frac{\hat{\xi}(\kv_1) - \hat{\xi}(\kv_2)}{k_1^2 + k_2^2 + \tx\frac{2}{m+1} \kv_1 \cdot \kv_2}.
\eeq

\n
One of the  result proven in \cite{CDFMT} is that there exists a critical value of the mass  $ \mstaroo $, approximately given by 
\beq
	\mstaroo \simeq 0.0735 = (13.607)^{-1},
\eeq
such that, if $ m > \mstaroo $  the quadratic form $ \formfo $ is closed and bounded from below (in fact positive) on its domain, so defining a self-adjoint operator $ H_0 $ which turns out to be the Friedrichs extension of $\tilde{H}_0$. On the other hand, for $ m < \mstaroo $  it is shown that the form is unbounded from below, which implies that the operator can not be at the same time self-adjoint and bounded from below (see, e.g., \cite[Proposition 4.1]{FT}).


\n
In the stable case $ m > \mstaroo $, $ H_0 $ is the following operator
\begin{eqnarray}
	\label{H0 domain}
	\dom(H_0) & = & \lf\{ \lf. u \in \ldf(\R^6) \: \ri| \: u = w + \pot \xi, w \in \dot{\hdf}(\R^6), \xi \in \dom_0, 
		 \int_{\R^3} \diff \kvp \: \hat{w}(\kvp,\kv) =  \lf( \widehat{\Gamma_0 \xi} \ri)(\kv) \ri\}		\nonumber	\\
	H_0 u & = & \hamfree w,	
\end{eqnarray}
where $ \hamfree $ is the free Laplacian in the center of mass coordinates, i.e., the multiplication operator by $ k_1^2 + k_2^2 + \tx\frac{2}{m+1} \kv_1 \cdot \kv_2 $ in Fourier transform, $ \Gamma_0 $ is  the self-adjoint operator in $L^2(\R^3)$ associated with the closed and positive quadratic form $ \Phi_0 $, i.e.,
\beq
	\label{Gamma0 intro}
	\lf( \widehat{\Gamma_0 \xi} \ri)(\pv) : = 2\pi^2  \sqrt{\frac{m(m+2)}{(m+1)^2}} \, p \: \hat{\xi}(\pv) + \int_{\R^3} \diff \qv  \: \frac{\hat{\xi}(\qv)}{p^2 + q^2 + \frac{2}{m+1} \pv \cdot \qv},
\eeq
and $ \dom_0 $ is its natural domain of self-adjointness. 
It is easy to verify that if $ u \in \dom(H_0) $ then the condition \eqref{plane asympt}  is satisfied for $\xi \in \dom_0$ and $\alpha=0$. Indeed, applying the Fourier transform, \eqref{plane asympt} can be translated into the boundary condition
\beq
	\int_{|k'| \leq N} \diff \kvp \: \hat{u}(\kv,\kvp) = N \: \xi(\kv) + o(1),	\qquad		\mbox{as } N \to \infty,
\eeq
which is   satisfied by any function in $ \dom(H_0) $. Such results extend to the case of $ N $ fermions of one species interacting with a different test particle, although in that case the condition on the mass for stability is not optimal.


At least at a numerical or heuristic level, it is however known, as discussed, e.g., in \cite{ ct,wc1},  
that there are other  possible extensions of the STM operator for $ \alpha = 0 $ (for the characterization of the  extensions in the case $\alpha \neq 0$ see \cite{m3,m4}). More precisely, there exists $\mssoo > \mstar$, approximately given by 
\begin{center}
	{$ \mssoo \simeq 0.116 = (8.62)^{-1}$},
\end{center}
such that for
\beq
	\mstaroo < m < \mssoo 
\eeq
there exists a family of  self-adjoint extensions of $\tilde{H}_0$ whose domains  are given by functions decomposing as in $ \dom(H_0) $ but with singular charges $\xi$ not belonging to $H^{1/2}(\R^3)$.  More precisely, their asymptotic behavior is characterized as follows
\be\label{bibi}
\hat{\xi}(\kv) = \tilde{\xi}_{n}(k) Y_{1}^n(\vartheta_k,\varphi_k) \,, \;\;\;\;\;\; n=-1, 0,1
\ee
 with
\beq
	\label{charge asympt}
	 \tilde{\xi}_{n}(k) \propto \frac{q}{k^{2-s}} + \frac{\beta q}{k^{2+s}} + o(k^{-2-s}), 	\qquad 	\mbox{as } k \to \infty, 
\eeq
where $Y_l^n$ denotes the spherical harmonics of indices $(l,n)$,  $ q $ is a complex number, $ \beta \in \R $ is a parameter labeling such operators and $ 0<s = s(m) <1  $ is another parameter depending on the mass $ m $ (see next Section for further details). Notice that the 3 dimensional Fourier anti-transform of $ k^{-2+s} $ does not belong to $ \dot{H}^{1/2}(\R^3) $ for any $ s > 0 $, since the function does not decay sufficiently fast as $ k \to \infty $. In terms of self-adjoint extensions the one studied in \cite{CDFMT} belongs to the family and is (formally) recoverd for $ \beta = +\infty $. Such an extension has indeed the smallest possible domain (Friedrichs extension), whereas the one with $ \beta = 0 $ show the largest domain (Krein extension).

\n
It is worth noticing that the parameter $s(m)$  is determined by requiring that a charge of the form $ k^{-2+s} \,Y_{1}^n(\vartheta_k,\varphi_k)$ is formally in the kernel of the operator \eqref{Gamma0 intro} (see eq. \eqref{exponent s} below). As we shall see, the existence of charges of the form \eqref{bibi}, \eqref{charge asympt} implies a further boundary condition satisfied by the wave function at the triple coincidence point $\yv_1=\yv_2=0$. Therefore, following the analogy with the two-body case, one can say that in the special case $\beta=0$ the Hamiltonian exhibits a ``three-body resonance''.

The aim of this paper is to give a rigorous construction of  such self-adjoint extensions. Following the line of \cite{CDFMT}, the method of the proof is again based on the theory of quadratic forms in Hilbert spaces. 
In Section 2  we give the precise formulation of the problem and state the main results.  
The proofs are postponed in Section 3. 
The Appendices collect some technical results used in the rest of the paper.

\vspace{1cm}
\n		
{\bf Acknowledgments}. This work was partially supported by the MIUR-FIRB grant 2012 ``Dispersive dynamics: Fourier analysis and variational methods", code RBFR12MXPO  (D.F.), the MIUR-FIR  grant 2013 ``Condensed Matter in Mathematical Physics", code RBFR13WAET (A.M. and M.C.), a INdAM 2014-2015 Progetto Giovani GNFM grant (A.M.), and a 2013-2014 ``CAS-LMU Research in Residence" grant (A.M.). Part of this work has been carried out during a visit of A.M. and M.C. at CIRM (Fondazione Bruno Kessler), Trento, funded by a 2013 Research in Pairs CIRM grant, as well as during a visit of all five authors at the Center for Advanced Studies at LMU Munich, funded by a ``CAS-LMU Research in Residence" grant.

\section{Main Results}
\label{sec: main results}

In this section we formulate the main results contained in the paper. First we introduce a suitable  quadratic form and prove its closedness. Next we derive the self-adjoint operator  associated with such form which turns out to be a self-adjoint extension of the STM operator. Finally the end of the Section is devoted to some comments and remarks.

\subsection{Quadratic form}

We introduce here the main object under investigation, namely the quadratic form $\formo$ and study its properties. We start by defining the critical masses which will play a crucial role in the following analysis.  For any 
$ s \in [0,1] $, we consider 
the equation
\beq
	\label{exponent s}
	 \pi \sqrt{\frac{m(m+2)}{(m+1)^2}} +  \int_{-1}^1 \diff t \: t \int_0^{\infty} \diff p \: \frac{p^{s}}{p^2 + 1 + \frac{2}{m+1} t p} = 0,
\eeq	
As we mentioned in the introduction, such an equation follows by imposing that the function $ k^{-2+s} \,Y_{1}^n(\vartheta_k,\varphi_k)$ is formally in the kernel of the operator \eqref{Gamma0 intro} and therefore it corresponds to a  ``three-body resonance'' condition.

\n
One can show (see proposition \ref{pro: critical masses}) that \eqref{exponent s}  has a unique solution $m(s)$, monotonically increasing in $s$. 
We call $ s(m) $ the inverse function of $ m(s) $, namely the unique solution of \eqref{exponent s} w.r.t. $ s $, for given $ m > 0 $. The most relevant quantities are introduced in the next definition. 
	
\begin{defi}[Critical masses]
	\label{defi: critical masses}
	\mbox{}	\\
	We define the critical masses $ \mstar < \mss $ as 
	\beq
		\label{critical masses}
		\mstar : = m(0) \simeq 0.0735,	\qquad 	\qquad		\mss : = m(1) \simeq 0.116.
	\eeq
	\end{defi}
	
	\n
From now on we fix the value of the mass  $ m $ in such a way that 
\beq
	\mstaroo < m < \mssoo.
\eeq
We  also denote  for short
\beq
	\label{betav}
	\betav : = \lf\{ \betal \ri\}_{n \in \{-1, 0, +1\}},
\eeq
with $ \betal \in \R \cup \{+ \infty\} $.  The quadratic form we are going to study is the following 

\begin{eqnarray}
 	\dom[\formo] &=& \lf\{ \lf. u \in \ldf(\R^6) \: \ri| \: u = w + \pot \eta, w \in \dot{\hf}(\R^6), \eta \in H^{-1/2}(\R^3), \eta = \xi + \sum_{n = -1}^{+1} q_{n} \xximinus, \ri.	\nonumber	\\
 	&& \lf. \;\;  \xi \in \dot{H}^{1/2}(\R^3), q_{n} \in \bC \ri\},	\label{form domain} \\
 	\formo[u] & = & \uform[w] + 2\qformo[\xi] +  \sum_{n = -1}^{+1} \betal \lf| q_{n} \ri|^2\label{form action}
\end{eqnarray}
with 
\be
\widehat{\xximinus} = \ximinus Y_1^n \,, \qquad \ximinus(k)= \f{1}{k^{2-s(m)}}
\ee
Before stating the main result it is worth discussing further the above expression of the quadratic form. First of all  the quadratic form discussed in \cite{CDFMT} is  obtained for $ \betal = +\infty $, which in turn implies $ q_n = 0 $ for all $ n $.  Moreover the quadratic form $ \formo $ is a perturbation of the quadratic form associated with the free Hamiltonian $ \hamfree $, which is supported on the coincidence planes $ \{ \yv_i = 0 \} $, i.e., a zero-range perturbation. 
Concerning the regularity class of the charges, we stress that the assumption $\eta \in H^{-1/2}(\R^3)$ is a necessary condition implied by the requirement 
\be
u=w+ \pot \eta \in \ldf(\R^3)\,, \;\;\;\;\;\; \text{with}\;\;\;\; w\in \dot{\hf}(\R^6).
\ee
This fact will be discussed in details in Appendix \ref{charge}. 
Notice that $k^{-1/2} \xximinus (\kv)$ is $L^2$-summable for large $k$ if and only if 
\beq
	s(m) < 1, \quad \mbox{   or equivalently,   } \quad m < \mss.
\eeq

\n
We also note that the assumption $\eta \in H^{-1/2}(\R^3)$ imposes a further constraint on the behavior of the charge $\hat{\xi}(\kv)$ only for $k$ small. More precisely, for $0<s(m) \leq 1/2$ the charge $\hat{\xi}$ must compensate the singularity at the origin of $\ximinus$ in order to have  $\hat{\eta} \in L^2_{\mathrm{loc}}(\R^3)$ (recall that $\eta \in H^{-1/2}(\R^3)$ implies $\hat{\eta} \in L^2_{\mathrm{loc}}(\R^3)$). For $1/2 < s(m) <1$ one has $\ximinus \in L^2_{\mathrm{loc}}(\R^3)$ and then also $\hat{\xi} \in L^2_{\mathrm{loc}}(\R^3)$. Therefore, in this case we have $\xi \in H^{1/2}(\R^3)$.

\n
The main result about the quadratic form $ \formo $ is formulated in the next theorem, proved in Section 3.1. 

\begin{teo}[Closedness of $ \formo $]
	\label{teo: closure}
	\mbox{}\\
	For any $ \mstaroo < m < \mssoo $ and $ \betav $, the quadratic form $ \formo $ is closed and bounded from below on the  domain $ \mathscr{D}[\formo] $. 
\end{teo}

	\begin{remark}[Boundedness from below of $ \formo $]
		\label{rem: boundedness}	
		\mbox{}	\\
		From the proof of Proposition \ref{pro: formol closure} it is clear that the quadratic form $ \formo$ is positive for $\beta_n  \geq 0$ for any $ n = -1,0,+1 $,  while if $\beta_n <0$ for some $ n $ a lower bound is explicitly given by 
		\be
			\label{eq: boundedness}
			E_0 (m) =- \left[2 (1-s(m) )\, \f{D_1 c_1 + D_2 (D_1 +1)}{D_1 D_2 c_1} \, \max_{n \in \{-1,0,+1 \}} |\beta_n|  \right]^{1/s(m)},
		\ee
		where $ c_1, D_1, D_2 $ are finite constants inherited from previous inequalities in the proof (see Section \ref{sec: closedness} and specifically \eqref{eq: d1 d2} and \eqref{pot norm 1}).
		Notice that $E_0 \to -\infty$ as $m\to m^\star$ and  $E_0 \to 0$ as $m\to m^{\star\star}$.
	\end{remark}

	\begin{remark}[Parameter $ s(m) $]
		\label{rem: s dependence}
		\mbox{}	\\
		By direct inspection of the proof one can realize that the parameter $s(m) $ can be replaced with any positive real number $ 0< s <1 $, i.e., not solving the algebraic equation \eqref{exponent s}, and the closedness of $ \formo $ would not be affected. On the other hand, the corresponding self-adjoint operator is an extension of the STM operator if and only if  $s=s(m)$ (see next Section \ref{22}).  
	\end{remark}

\subsection{Self-adjoint extensions of the STM operator}\label{22}

\noindent
We are now able to introduce the operator associated with $ \formo $. We set
\begin{eqnarray}
	\label{operator domain}
	\mathscr{D}\lf(\hamo\ri) & = & \lf\{ \lf. u \in \ldf(\R^6) \: \ri| \: u = w + \pot \eta, w \in \dot{\hdf}(\R^6), \eta \in H^{-1/2}(\R^3), \eta = \xi + \sum_{n = -1}^{+1} q_{n} \xximinus, \ri.		\nonumber \\
	&&	\lf. \Gamma_0 \xi \in L^2(\R^3), q_{n} \in \bC,	\int_{\R^3} \diff \kvp \: \hat{w}(\kvp,\kv) =  \lf( \widehat{\Gamma_0 \xi} \ri)(\kv), \ri.	\nonumber	\\
	&& \lf. \beta_{n} q_{n} = 2 \lim_{\eps \to 0} \lf( \xximinus, {\Gamma_0 \xi} \ri)_{L^2(\R_{\eps}^3)} \ri\}	\nonumber\\
	\hamo u & = & \hamfree w,	
\end{eqnarray}
where $ \Gamma_0 \xi$ is given by  the following expression
\beq
	\label{Gamma0}
	\lf( \widehat{\Gamma_0 \xi} \ri)(\pv) : = 2\pi^2  \sqrt{\frac{m(m+2)}{(m+1)^2}} p \: \hat{\xi}(\pv) + \int_{\R^3} \diff \qv  \: \frac{\hat{\xi}(\qv)}{p^2 + q^2 + \frac{2}{m+1} \pv \cdot \qv},
\eeq
and we have denoted for short
\beq
	\label{rdeps}
	\R^d_{\eps} : = \lf\{ \kv \in \R^d \: \big| \: k \geq \eps \ri\}.
\eeq
In the previous definitions $ \Gamma_0 $ has to be understood as the formal action of the integral operator \eqref{Gamma0}, without any reference to its counterpart as operator on a Hilbert space.
Note however that if $ \Gamma_0 $ given in \eqref{Gamma0} is properly restricted to its maximal domain within the Hilbert space $ L^2(\R^3)$, then it coincides with positive, self-adjoint operator associated with the quadratic form $\qformo$ with maximal domain (Friedrichs extension). 
 A closer inspection on the condition $ \Gamma_0 \xi \in L^2(\R^3) $ reveals that it is certainly satisfied by $ H^1-$functions since both the diagonal and off-diagonal term in \eqref{Gamma0} can be easily bounded by the $H^1-$norm of $ \xi $. However,  there exist $ H^{1/2}-$functions $\xi $ so that both the diagonal and off-diagonal terms are not in $ L^2(\R^3) $ but their sum, i.e., $ \Gamma_0 \xi $, is finite almost everywhere and defines a function in $ L^2(\R^3) $. 
Concerning the operator $\hamo$, in Section 3.2  we prove the following theorem. 

\begin{teo}[Self-adjointness of $ \hamo $]
	\label{teo: hamo}
	\mbox{}	\\
For any  $ \mstaroo < m < \mssoo $ and $ \betav $,	the operator $\hamo$ with domain $ \mathscr{D}(\hamo) $ is self-adjoint and bounded from below, and its quadratic form is $\formo$.  
\end{teo}

\begin{remark}[Comparison with \cite{m2}]
	\mbox{}	\\
	As already pointed out, in \cite{m2} it was studied the case $ \alpha \in \R $ by an operator theoretical approach: by studying the deficiency indeces of the STM operator, it is proven that the STM operator admits a one parameter family of self-adjoint extensions whenever
	\beq
		m < \tilde{m}^{\star\star},		\qquad		\mbox{with } \tilde{m}^{\star\star} \simeq 0.0812
	\eeq
	A direct comparison with our result immediately shows that the region above $ \mstar $ where such self-adjoint extensions exist is narrower, since trivially 
	\beq
		\tilde{m}^{\star\star} < \mssoo. 
	\eeq
	It is important to remark however that the threshold $ \mssoo $ is the one obtained by heuristic arguments in the physics literature (see \cite[p. 45]{wth} or \cite{E}) for the unitary regime $ \alpha = 0 $.
\end{remark}

\begin{remark}[Operator theoretical approach]
	\mbox{}	\\
	The possible discrepancy with the result obtained in \cite{m2} motivates a comment about the differences in the approach to the problem. Indeed the work \cite{m2} is based on a standard operator theoretical analysis of the deficiency indeces of the STM operator for $ \alpha \in \R $. More precisely the deficiency spaces turn out to be subspaces of the space of charges $ \eta $, living on the coincidence planes $ \{ \xv_i = \xv_j \} $. Such methods apply as well to the case $ \alpha = 0 $. However what makes the extension non trivial is that, while for $ \alpha \neq 0 $ the condition $ \eta \in L^2(\R^3) $ is required and the usual theory of operators in Hilbert spaces can be applied, on the opposite for $ \alpha = 0 $ the space of charges is much larger, i.e., $ \eta \in H^{-1/2}(\R^3) $ (see Appendix \ref{charge}). At the level of quadratic forms this can be easily understood by observing that if $ \alpha \neq 0 $ an additional term proportional to $ \alpha \lf\| \eta \ri\|^2_{L^2(\R^3)} $ arises and therefore the assumption $ \eta \in L^2(\R^3) $ can not be avoided to obtain a meaningful expression. Obviously for $ \alpha = 0 $ such a term is absent and the only constraint to $ \eta $ is $ \pot \eta \in L^2_{\mathrm{loc}}(\R^6) $, which yields the condition $ \eta \in H^{-1/2}(\R^3) $ as proven in Appendix \ref{charge}.
\end{remark}

\n
In the next proposition  we establish the relation between   $\hamo$ and  the STM operator. 

\begin{pro}[STM extensions]
	\label{pro: STM}
	\mbox{}	\\
	The operator $ \hamo $ with domain $ \dom(\hamo) $ is a self-adjoint extension of the STM operator $ \tilde{H}_0 $. 
	\end{pro}


\n
The domain of $ \hamo $ contains also charges with a singular behavior for $k\rightarrow \infty$.  More precisely, denoting

	\beq
		\nu(m) :=  \frac{8 \pi}{m+1} \disp\int_{1}^{\infty} \diff p \: \frac{1}{p}\disp\int_{0}^1 \diff t \: t^2 \int^{1/p}_{0} \diff q \frac{q^{1-s(m)}}{(q^2 +1)^2 - \frac{4}{(m+1)^2} q^2},	
			\eeq
  we prove the following result.
	
	\begin{pro}[Charge asymptotics]
		\label{pro: charge asympt}
		\mbox{}	\\
		The domain $ \dom(\hamo) $ (recall \eqref{betav} and \eqref{operator domain}) contains charges $ \eta(\kv) $ such that
		\bdn
			\hat{\eta}(\kv) &=& \sum_{n = - 1}^{1}  \tilde{\eta}_{n}(k) Y_{1}^n(\vartheta_k, \varphi_k),	\nonumber	\\
			 \tilde{\eta}_{n}(k) &\underset{k \to \infty}{\sim}& \frac{q_{n}}{k^{2 - s(m)}} + 	\frac{r_n}{ k^{2 + s(m)}} + o(k^{-2-s(m)})\,, \\
			r_n&= & \frac{\beta_{n} }{\nu(m)} \, q_n.
			\label{fbc}	
		\edn	
	\end{pro}
	
\n
The above propositions will be proved in Section 3.2. Here we notice that  \eqref{fbc} represents a further boundary condition (see the analogy with \eqref{bcde}) satisfied by the elements of $\dom(\hamo)$, besides the standard boundary condition characterizing the STM operator	\eqref{plane asympt}. Condition \eqref{fbc} is equivalent to the boundary condition known in the physical literature characterizing the behavior of the wave function at short distances (see, e.g., \cite{wc1}). 

Let us conclude this Section on our main results with a few additional remarks that clarify the relevance of our work.
	
	\begin{remark}[Boundedness from below of $ \hamo $]
		\mbox{}	\\
		The lower bound discussed in Remark \ref{rem: boundedness} clearly applies to $ \hamo $ too, and therefore $ \hamo $ is a positive operator if $ \beta_n \geq 0 $, whereas it is semi-bounded if some $ \beta_n < 0 $, with the r.h.s. of \eqref{eq: boundedness} providing a bound from below.
	\end{remark}
			
	
	\begin{remark}[Parameter $ s(m) $ and STM extensions]
      	\label{rem: s and STM}
      	 \mbox{} \\
        	As anticipated in the Remark \ref{rem: s dependence} the specific choice $ s = s(m) $ is crucial only to capture STM-type operators. The above Proposition indeed guarantees that, if such a choice is made, $ \hamo $ extends the STM operator. However one can also realize that the converse is true, namely if $ s \neq s(m) $ then there is still a self-adjoint operator $ \hamo $ which however {\it is not} an extension of $ \tilde{H}_0 $. A simple way to prove this is by looking at the first boundary condition in \eqref{operator domain}: whenever $ s = s(m) $, the formal action of $ \Gamma_0 $ on $ \xximinus $ identically vanishes (see \eqref{gamma0 vanishes}), i.e., $ \Gamma_0 \xximinus = 0 $, so that the condition can be rewritten
        	\beq
                \label{eq: bc STM alternative}
                \int_{\R^3} \diff \kvp \: \hat{w}(\kvp,\kv) =  \lf( \widehat{\Gamma_0 \eta} \ri)(\kv),
        	\eeq
        	i.e., on the r.h.s. $ \xi $ can be replaced with $ \eta $, the full charge. This qualifies the STM extensions, since if it applies, any wave function in $ \dom(\hamo) $ has the STM behavior \eqref{plane asympt} on the planes where the interaction is supported, i.e.,
        	\beq
        		\label{eq: bc spatial}
                \Psi(\yv_1,\yv_2) =  \frac{(-1)^{i+1}\eta(\yv_j)}{4 \pi | \yv_i|} + \alpha (-1)^{i+1}\eta(\yv_j) + o(1),        \qquad  \mbox{as } | \yv_i| \to 0, \: i,j = 1,2, \, i \neq j.
        	\eeq
        	Obviously the same {\it does not} hold for a generic $ s \neq s(m) $, since the boundary condition  {\it can not} be cast in the \eqref{eq: bc STM alternative} form.
	\end{remark}
	
	\begin{remark}[Self-adjoint extension above $ \mssoo $]
		\mbox{}	\\
		As discussed in the previous Remark the self-adjoint extensions considered in this paper exist only for $ m < \mssoo $. For $ m > \mssoo $ no such extension exists and the STM operator is in fact essentially self-adjoint. The  unique self-adjoint extension is for $ m > \mssoo $ the one studied in \cite{CDFMT}, or the one corresponding formally to $ \beta = + \infty $: by taking $ \beta = + \infty $ we mean that $ q = 0 $ for any charge $ \eta $ and therefore functions in the domain of the operator have no singularity at the triple point $ \lf\{ \xv_1 = \xv_2 = 0 \ri\} $ (see next Remark \ref{rem: triple point}).
	\end{remark}
	
	\begin{remark}[Asymptotics at the triple coincidence point]
		\label{rem: triple point}
		\mbox{}	\\
		The second boundary condition in \eqref{operator domain} can be interpreted also as a prescription about the asymptotic behavior of singularities in position space at $\yv_1= \yv_2=0$, i.e., when the positions of the three particles coincide.  
Indeed, if we look at the behavior of the potentials generated by $\ximinus $ and $\xiplus$,  we have by scaling
		\[
			\pot (\ximinus Y^n_1) (\mu \yv_1 , \mu \yv_2)= \f{1}{\mu^{2+s} } \pot (\ximinus Y^n_1) ( \yv_1 , \yv_2), \qquad  
			\pot (\xiplus Y^n_1) (\mu \yv_1 , \mu \yv_2)=  \f{1}{\mu^{2-s} }  \pot (\xiplus Y^n_1) ( \yv_1 , \yv_2).
		\]
		Therefore we have
		\beq
			\label{eq: triple point}
			\pot \eta   (\mu \hat{\yv}_1 , \mu \hat{\yv}_2) \underset{\mu\to 0}{\sim} \f{1}{\mu^2} \lf( g_1(m)\f{q}{\mu^s} + g_2(m) \frac{\beta q}{\nu(m)}  {\mu^s}  + o(\mu^s) \ri),
		\eeq
		where we have denoted by $ \hat{\yv}_{i} $ any unit vector in $ \R^3 $, i.e., such that $ |\hat{\yv}_i| = 1 $, and\footnote{Note that by rotational invariance $ g_{1,2} $ are both independent of $ \hat\yv_{1,2} $.}
		\beq
			g_{1,2}(m) = \frac{1}{(2\pi)^3} \int_{\R^3} \diff \kv_1 \diff \kv_2 \: \frac{e^{i \kv_1 \cdot \hat\yv_1 + i \kv_2 \cdot \hat\yv_2}}{k_1^2 + k_2^2 + \frac{2}{m+1} \kv_1 \cdot \kv_2} \lf\{ \frac{Y_{1,n}(\vartheta_1,\varphi_1)}{k_1^{2 \pm s(m)}}  - \frac{Y_{1,n}(\vartheta_2,\varphi_2)}{k_2^{2 \pm s(m)}}  \ri\}.
		\eeq
		We see that the coefficients of the two more singular terms must be proportional as is customary stated in the physical literature \cite[note 43]{wc1}.
	\end{remark}
	
	\begin{remark}[Further extensions for $ \ell > 1 $]
		\mbox{}	\\
		A priori one could imagine of reproducing the analysis performed here within any sector with angular momentum $ \ell > 1 $ odd, and therefore costruct other families of self-adjoint extensions of the STM operator with a larger domain in the subspace of charge with angular momentum $  \ell  > 1 $. Obviously this would introduce further threshlolds $ \mssol $ depending on $ \ell $ for the existence of such extensions, which are allowed only for $ m < \mssol $. However the analysis contained in Appendix \ref{sec: critical masses} shows that $ \mssol $ is a decreasing function of $ \ell $ (see Lemma \ref{lemmafunzione}), i.e., $ \mssol \leq m_3^{\star\star} $ for any $ \ell \geq 3 $ odd, and already for $ \ell = 3 $ one has \eqref{m3star}
		\beq
			m^{\star\star}_3 < \mstar.
		\eeq
		Hence all those possible extensions are certainly unbounded from below. This is why we do not investigate this question further.
	\end{remark}
	
		\n
		It is worth to comment further on the peculiar structure of the 3-body Hamiltonians constructed and discussed in this work, which is typical for self-adjoint extensions of bounded-below symmetric operators (in our case, the operator $ \tilde{H}_0 $ we started with). Our findings first of all recover the natural self-adjoint Hamiltonian of STM type $ H_{0,\infty} $ (that is, $ \beta=\infty $ in the present notation): this is precisely the Hamiltonian associated with the quadratic form $ \F_0 $ defined in \eqref{formfo domain}-\eqref{qformo}. Next to it, for each fixed $ m $ in the considered regime, all other Hamiltonians $ \hamo $ have been obtained by enlarging the domain of the quadratic form $ \F_0 $, as done in \eqref{form domain}-\eqref{form action}, and this enlargement consisted of adding each time a one-dimensional subspace of charges spanned by a function that has, radially in the momentum coordinate, the singularity $ \xi^-(k)=k^{-2+s(m)} $. This is a general fact for the quadratic form of each extension of a semi-bounded symmetric operator, that is, such a form is defined also on a suitable additional one-dimensional subspace of the defect subspace of the initial operator (see, e.g., \cite{simon,bi}).
		
In our case, we see that the singular behaviour of the extra charge qualifies the domain of the extended quadratic form. We then demonstrated that for each such enlarged forms the domain of the corresponding self-adjoint Hamiltonian not only reproduces the standard STM condition at the coincidence hyperplane, but it is further characterised by an additional relation between the extra (singular) charge and the regular charge -- this is the second condition in \eqref{operator domain}. Last, and physically most relevant, we cast this constraint on the singular charge in a form \eqref{fbc} that, when applicable, has a natural interpretation as boundary condition of the three-body wave function at the triple coincidence point \eqref{eq: triple point}.

In short we may say that, fixed $ m $, each Hamiltonian of the class considered in this work is eventually characterised by the worse (most singular) behaviour that the wave-functions of its domain undergo in the vicinity of the triple coincidence point.

\section{Proofs}

\subsection{Closedness of the quadratic form}
\label{sec: closedness}

This Section is devoted to the proof of the closedness of the quadratic form $ \formo $ on the domain $ \dom[\formo] $ given by \eqref{form domain}. 

\n
The first trivial but useful observation it that the quadratic form $ \qformo $ acting on the charge is invariant under rotations, i.e., it is block diagonal in the subspace decomposition of the Hilbert space induced by eigenvectors of the angular momentum. More precisely, let us introduce the subspaces

\be
\hilb_{\ell} = \bigg\{ \eta \in L^2(\R^3) \: | \: \hat{\eta}(\kv) = \sum_{n=-\ell}^{\ell} \tilde{\eta}_{\ell,n}(k) Y^{n}_{\ell}(\vartheta_k, \varphi_k) \bigg\}
\ee
and
\beq
	\hilb_{\ell,n} = \left\{ \eta \in L^2(\R^3) \: | \: \hat{\eta}(\kv) = \tilde{\eta}_{\ell,n}(k) Y^{n}_{\ell}(\vartheta_k, \varphi_k) \right\}.
\eeq
We can thus consider separately the closure of the form restricted to any given subspace $ \hilb_{\ell, n} $ of the charge space. 
Indeed, for any admissible charge $ \eta $ in $ \dom\lf[\formo\ri] $, the decomposition
\beq
	\label{charge spherical}
	 \hat\eta(\kv) = \sum_{\ell \in \N} \sum_{n = -\ell}^{\ell} \tilde\eta_{\ell, n}(k) Y_{\ell}^n(\vartheta_k,\varphi_k),
\eeq
implies the splitting \cite[Lemma 4.2]{CDFMT}
\beq
	\qformo[\xi] = \sum_{\ell \in \N} \sum_{n = -\ell}^{\ell} \qformo[\tilde\xi_{\ell, n}],
\eeq
where (recall that $ P_{\ell} $ is the $\ell-$th Legendre polynomial)
\bml{
	\qformo[\tilde\xi_{\ell, n}] : = 2\pi^2 \frac{\sqrt{m(m+2)}}{m+1} \int_{0}^{\infty} \diff p \: p^3 \: |\tilde\xi_{\ell, n}(p)|^2 	\\
					+ 2 \pi \int_{-1}^1 \diff t \: P_{\ell}(t) \int_{0}^{\infty} \diff p \: p^2 \int_0^{\infty} \diff  q \: q^2 \frac{\tilde\xi_{\ell, n}^*(p) \tilde\xi_{\ell, n}(q)}{p^2 + q^2 + \frac{2}{m+1} p q t}. \label{qformol}
}
Note that the angular momentum decomposition \eqref{charge spherical} is not the one associated with the whole wave function $ u $ but the one inherited from the function $ \eta $ of the singular term $ \pot \eta $. 

On the other hand it is obvious from the definition that for any $ u \in \dom[\formo] $ such that $ u = w + \pot \eta $, with 
\beq
	\eta \in \lf( \bigoplus_{\ell \neq 1} \hilb_{\ell} \ri) \bigcap H^{-1/2}(\R^3),
\eeq
then
\beq
	\formo[u] = \F_0[u],
\eeq
i.e., the only subspace on which the quadratic form $ \formo $ differs from $ \F_0 $ is the one identified by the condition $ \eta \in \hilb_1 $. We thus restrict our analysis to the quadratic forms
\begin{eqnarray}
 	\dom\lf[\formol \ri] &=& \lf\{ \lf. u \in \dom[\formo] \: \ri| \: \eta \in \hilb_{1,n} \cap H^{-1/2}(\R^3), \eta = \xi + q \xximinus, \xi \in \dot{H}^{1/2}(\R^3), q \in \bC \ri\},	\label{formol domain} \\
 	\formol[u] & = & \uform[w] + 2\qformo[\xi] + \beta \lf| q \ri|^2. \label{formol action}
\end{eqnarray} 
where $ n = - 1, 0, +1 $ and $ \beta \in \R $.
Note that the $ n-$dependence of $ \formol $ is in fact trivial, namely it appears only in the domain decomposition whereas the action of both forms $ \formol $ and $ \qformo^{(1)} $ does not depend on $ n $. Hence if one proves closedness of each form $ \formol $ separately, Theorem \ref{teo: closure} is proven as well. From now on we will simplify the notation by redefining 
	\beq
		\label{notation abuse 1}
		\qformo : = \qformo^{(1)},		\qquad		\Xi^- : = \xximinus.
	\eeq 


	\begin{pro}[Closedness of $ \formol $]
		\label{pro: formol closure}
		\mbox{}	\\
		For any $ n = - 1, 0, +1$, $ \beta \in \R $ and $ \mstar < m < \mss $, the quadratic form $ \formol $ is closed and bounded from below on $ \dom[\formol] $.
	\end{pro}

	\begin{proof}	
		The proof is slightly simpler in the case $ \beta \geq 0 $, that we will consider first, since boundedness from below is much easier to show. All the terms of the quadratic form are indeed positive in this case and the only one whose positivity is non-trivial is $ \qformo $, but the assumption $ m > \mstar $ directly implies $ \qformo \geq 0 $ (see \cite[Proposition 3.1]{CDFMT}).
		
		Let us then consider a sequence $u_i \in \mathscr{D}[\formol] $ such that $\formol [u_i - u_j] \to 0$ for $ i,j \rightarrow \infty$ and $u_i \rightarrow u$ in $L^2_{\mathrm{f}}(\R^6)$ for $  i \rightarrow \infty$.  Since $ u_i \in \mathscr{D}[\formol] $, there exist $ w_i \in \dot{\hf}(\R^6) $, $ \xi_i \in \dot{H}^{1/2}(\R^3) $ and $ q_i \in \mathbb{C} $ such that
		\beq
			u_i = w_i + \pot \lf(\xi_i + q_i \Xi^- \ri).
		\eeq
		Moreover using the lower bound in \cite[Eq. (3.13)]{CDFMT} for $ \qformo $, we have
		\bml{ \label{312}
			\formol [u_i -u_j] \geq \F_0[w_i - w_j] +2 \qformo[\xi_i - \xi_j] \\
			\geq D_1  \int_{\R^6} \!\diff \kv_1 \diff \kv_2 \lf( k_1^2 + k_2^2  \ri) \lf|						(\hat{w}_i-\hat{w}_j)(\kv_1,\kv_2) \ri|^2  
			+ D_2 \, \disp\int_{\R^3} \!\!\diff \pv 				\, p \, \lf|(\hat{\xi}_i - \hat{\xi}_j)(\pv)\ri|^2  }
		\be
			\label{eq: d1 d2}
			D_1=\f{m}{m+1} \, ,	\qquad		 D_2= 4\pi^2  \tx\frac{\sqrt{m(m+2)}}{m+1} \,(1 - \Lambda_1(m)) 
		\ee
		where $ \Lambda_{1}(m) < 1 $  as long as $ m  > \mstar $. Hence
		\be\label{cauwn}
			\int_{\R^6} \!\diff \kv_1 \diff \kv_2 \lf( k_1^2 + k_2^2  \ri) |(\hat{w}_i-\hat{w}_j)(\kv_1,\kv_2)|^2 \underset{i,j \to \infty}{\longrightarrow} 0
		\ee
		and
		\be\label{cauxin}
			\int_{\R^3} \!\!\diff \pv \, p \, |(\hat{\xi}_i - \hat{\xi}_j)(\pv)|^2 \underset{i,j \to \infty}{\longrightarrow} 0.
		\ee 
		Due to (\ref{cauwn}), there exists $\hat{g} \in   L^2_{\mathrm{f}}(\R^6)$ such that
		\be
			\int_{\R^6} \!\diff \kv_1 \diff \kv_2 \left| \sqrt{k_1^2 + k_2^2} \,\hat{w}_i (\kv_1,\kv_2) - \hat{g} (\kv_1,\kv_2) \right|^2 \underset{i \to \infty}{\longrightarrow} 0
		\ee
		and for any $\eps >0$
		\be\label{conwn}
			\int_{\R^6_{\eps}} \!\diff \kv_1 \diff \kv_2 \left| \hat{w}_i (\kv_1,\kv_2) - \hat{w}(\kv_1,\kv_2) \right|^2 \underset{i \to \infty}{\longrightarrow} 0,
		\ee
		where $\R^d_{\eps} $ is defined in \eqref{rdeps} and $ w \in \dot{H}^1(\R^6) $ with
		\be
			\hat{w}(\kv_1,\kv_2)= \f{\hat{g} (\kv_1,\kv_2)}{ \sqrt{k_1^2 + k_2^2}}. 
		\ee
 		Analogously, from (\ref{cauxin}) there exists $\hat{\nu} \in L^2(\R^3)$ such that 
		\be\label{conxi}
			\int_{\R^3} \!\!\diff \pv \,  \left| \sqrt{p}\, \hat{\xi}_i (\pv) - \hat{\nu}(\pv)\right|^2 					\underset{i \to \infty}{\longrightarrow} 0,	\qquad
			\int_{\R^3_{\eps}} \!\!\diff \pv \, \left| \hat{\xi}_i (\pv) -\hat{\xi}(\pv) \right|^2 \underset{i \to \infty}{\longrightarrow} 0,
		\ee
		where $ \xi \in \dot{H}^{1/2}(\R^3) $ and
		\be
			\hat{\xi}(\pv)= \f{\hat{\nu}(\pv)}{\sqrt{p}}.
		\ee
		Notice that (\ref{conxi}) also implies
		\be
			\int_{\R^6_{\eps}} \!\diff \kv_1 \diff \kv_2 \left| \widehat{\pot \xi}_i(\kv_1,\kv_2) - \widehat{\pot \xi}(\kv_1,\kv_2) \right|^2 \underset{i \to \infty}{\longrightarrow} 0.
		\ee
		From (\ref{conwn}), (\ref{conxi}) and the $L^2-$convergence of $u_i$ to $u$, we have convergence of $\pot (q_i \Xi^-)$ in $L^2(\R^6_{\eps})$. Hence $q_i$ is a Cauchy sequence in $\bC$ and there exists $q \in \bC$ such that $q_i \rightarrow q$ and $\pot (q_i \Xi^- ) \rightarrow \pot (q\Xi^- )$ in $L^2(\R^6_{\eps})$ for $i \rightarrow \infty$. Thus we obtain 
		\be
			u=w+\pot \left( \xi +q \Xi^-\right),
		\ee 
		which means $u \in  \mathscr{D}[\formol] $. It is now straightforward to verify that $\formol [u_i ] \rightarrow \formol[u]$ as $i \rightarrow \infty$, concluding this first part of the proof.
		
		In other case $ \beta < 0 $ we have to be more careful in bounding from below the form. For any $\la >0$, from \eqref{312} we have
\bml{ \label{st1}
			\formol [u] +\la \|u\|^2_{L^2(\R^6)} 		\\	\geq (D_1 -\delta_1) \!\!\int_{\R^6} \!\diff \kv_1 \diff \kv_2 \lf( k_1^2 + k_2^2  \ri) \lf|						\hat{w}(\kv_1,\kv_2) \ri|^2  
			+ (D_2 -\delta_2) \!\! \disp\int_{\R^3} \!\!\diff \pv 				\, p \, \lf|\hat{\xi}(\pv)\ri|^2  	 +Q_{\la} (u)
		}
		where
		\be
		Q_{\la}(u) :=\delta_1 \!\!\int_{\R^6} \!\diff \kv_1 \diff \kv_2 \lf( k_1^2 + k_2^2  \ri) \lf|						\hat{w}(\kv_1,\kv_2) \ri|^2  
			+ \delta_2 \!\! \disp\int_{\R^3} \!\!\diff \pv 				\, p \, \lf|\hat{\xi}(\pv)\ri|^2 +\la \|w+\pot \eta\|^2_{L^2(\R^6)} +\beta |q|^2  
		\ee
		and $\delta_i \in (0,D_i)$, $i=1,2$.  For the estimate of $Q_{\la}$ we first notice that
		\be
		\int_{\R^6} \!\diff \kv_1 \diff \kv_2 \lf( k_1^2 + k_2^2  \ri) \lf|						\hat{w}(\kv_1,\kv_2) \ri|^2  \geq \la \int_{k_1^2+k_2^2 >\la} \!\diff \kv_1 \diff \kv_2  \lf|						\hat{w}(\kv_1,\kv_2) \ri|^2 = : \la \, \lf\|w\ri\|^2_{\la, 1}
		\ee
		and
		\be
		\disp\int_{\R^3} \!\!\diff \pv 				\, p \, \lf|\hat{\xi}(\pv)\ri|^2 \geq \la \disp\int_{p>\sqrt{\la}} \!\!\diff \pv \,\f{1}{p}				 \, \lf|\hat{\xi}(\pv)\ri|^2 =: \la \, \lf\|\xi\ri\|^2_{\la,2}.
		\ee
	Therefore
	\bml{
	Q_{\la}(u) \geq \la \bigg( \delta_1  \|w\|^2_{\la1} + \delta_2  \|\xi\|^2_{\la,2} + \|w +\pot \eta\|^2_{\la,1} \bigg) + \beta |q|^2 \\
	\geq \la \bigg[    (\delta_1 +1) \|w\|^2_{\la,1} - 2 \|w\|_{\la,1} \|\pot \eta\|_{\la,1} + \|\pot \eta\|^2_{\la,1}   +\delta_2 \|\xi\|^2_{\la,2}  \bigg] + \beta |q|^2 \\
	= \la \bigg[    \left(\sqrt{\delta_1 +1} \|w\|_{\la,1} - \f{1}{\sqrt{\delta_1 +1}} \|\pot \eta\|_{\la,1}\right)^2 +  \f{\delta_1}{\delta_1 +1} \|\pot \eta\|^2_{\la,1}   +\delta_2 \|\xi\|^2_{\la,2}  \bigg] + \beta |q|^2 \\
	\geq 
	\la \bigg(      \f{\delta_1}{\delta_1 +1} \|\pot \eta\|^2_{\la,1}   +\delta_2 \|\xi\|^2_{\la,2}  \bigg) + \beta |q|^2
	}	
	where we have used the inequality $\|a +b\|^2 \geq \|a\|^2 +\|b\|^2 	-2 \|a\| \|b\|$. Using the lower bound in Lemma \ref{lem: pot norm 1} we have
		\be
		 \|\pot \eta\|^2_{\la,1} \geq c_1 \|\xi + q \Xi^- \|^2_{\la,2} \geq 
		 c_1 \|\xi\|^2_{\la,2} + c_1 |q|^2 \| \Xi ^- \|^2_{\la,2} -2 c_1 |q| \|\xi\|_{\la,2} \|\Xi^-\|_{\la,2}. 
		\ee
		Then, proceeding as before, we find
		\bml{\label{Q29}
		Q_{\la}(u) \geq \la \bigg[ \left( \f{\delta_1 c_1}{\delta_1 +1} + \delta_2 \right) \|\xi\|^2_{\la,2} 
		-2 \f{\delta_1 c_1}{\delta_1 +1} |q| \| \xi \|_{\la,2} \| \Xi^-\|_{\la,2} + \f{\delta_1 c_1}{\delta_1 +1} |q|^2 \|\Xi^-\|^2_{\la,2} \bigg] + \beta |q|^2\\
		\geq \la \, \f{\delta_1 \delta_2 c_1}{\delta_1 c_1 + \delta_2 (\delta_1 + 1)} |q|^2 \|\Xi^-\|^2_{\la,2} + \beta |q|^2.
		}
		An explicit computation of $\|\Xi^-\|^2_{\la,2}$ yields
		\be\label{ex30}
		\|\Xi^-\|^2_{\la,2} = \int_{\sqrt{\la}}^{\infty} \!\! \diff k \, k^{-3 + 2 s(m)} = \f{\la^{-1 + s(m)}}{2(1-s(m))}.
		\ee
		Inserting \eqref{ex30} in \eqref{Q29}, we conclude that there exists $\lambda_0 >0$ such that $Q_{\la}(u) >0$ for any  $\la >\la_0$.  Taking into account \eqref{st1}, this implies that the form is lower bounded. Furthermore, proceeding as in the case $\beta \geq 0$, one can easily show that  the form is also closed, concluding the proof. 
		\end{proof}

\subsection{Self-adjoint extensions}
\label{sec: sa extensions}

We now construct the self-adjoint operator associated with the quadratic form $ \formo $. Thanks to the
rotational invariance, 
we can restrict to one single form $ \formol $, for $ n = -1, 0, +1 $ and find the associated self-adjoint operator 
$  \hamol$.


	\begin{pro}[Self-adjointness of $ \hamo $]
		\label{pro: hamol}
		\mbox{}	\\
		For any $ n = - 1, 0, +1 $, $ \beta \in \R $ and $ \mstar < m < \mss $, the unique self-adjoint operator $ \hamol $ associated with the closed and bounded from below quadratic form $ \formol $ coincides with the restriction of $ \hamo $ to $ \hilb_n $.
	\end{pro}
	
	\begin{proof}
By quadratic form theory
the domain of the operator $  \hamol$ corresponding to $ \formol $ is given by 
		\beq
			\mathscr{D}\lf(\hamol\ri) = \lf\{ \lf. u \in \dom\lf[\formol\ri] \: \ri| \: \exists \psi \in L^2(\R^6), \formol[v, u] = (v, \psi), \forall v \in \dom\lf[\formol\ri] \ri\},
		\eeq
		and moreover $ \psi = : \hamol u $.

		We notice that any $ u \in \mathscr{D}\lf(\hamol\ri) $ must decompose as
		\beq
			\label{u decomposition}
			u = w_u + \pot \eta_u,	\qquad		\eta_u = \xi_u + q_u \Xi^-
		\eeq
		with $ w \in \dot{H}_{\rm f}^1(\R^6) $, $ \eta_u \in \hilb_{1,n} $ and $ q_u \in \bC $. Now we characterize 
the domain $ \mathscr{D}\lf(\hamol\ri) $, by picking some special subclasses of vectors $ v \in \dom[\formol] $:
		\begin{enumerate}[(a)]
			\item Let us first pick some $ v \in \hf(\R^6) $, i.e., a wave function with no singular part. In this case we obtain for any $ u $ in the domain $ \dom\lf(\hamol\ri) $
				\bml{
					\label{identity dom 1}
		 			\formol[v,u] = F[v,w_u] = \int_{\R^6} \diff \kv_1 \diff \kv_2 \: \lf( k_1^2 + k_2^2 + \tx\frac{2}{m+1} \kv_1 \cdot \kv_2 \ri) \hat{v}^*(\kv_1,\kv_2) \hat{w}_u (\kv_1,\kv_2)	\\
		 			 = \lf(v, \psi\ri)_{\ldf(\R^6)},
		 		}
		 		for some $ \psi \in \ldf(\R^6) $. 
Since this is finite for any $v$ it implies that
		 		\bdm
		 			\lf( k_1^2 + k_2^2 + \tx\frac{2}{m+1} \kv_1 \cdot \kv_2 \ri) \hat{w}_u (\kv_1,\kv_2) \in \ldf(\R^6),
		 		\edm
		 		i.e., $ w_u \in \dot{H}^2_{\rm f}(\R^6) $. Moreover a straightforward consequence is that
		 		\beq
		 			\label{hamol action}
		 			\hamol u = \psi = \hamfree w_u.
		 		\eeq
		 		
			\item 	Now we take
					\bdm
						v = w_v + \pot \xi_v,
					\edm
					for some
$ \xi_v \in H^{1/2}(\R^3) \cap \hilb_{1,n} $, and  we obtain via \eqref{hamol action}
					\beq
						\label{identity dom 2}
						\formol[v,u] = F[w_v,w_u] + 2\qformo[\xi_v,\xi_u] = \lf(v, \psi\ri)_{\ldf(\R^6)} = \lf(w_v + \pot \xi_v, \hamfree w_u \ri)_{\ldf(\R^6)}.
					\eeq
					Now, even if their sum belongs  to $ L^2(\R^6) $ neither $ w_v $ nor $ \pot \xi_v $ does. Still, since $$ \lf(w_v, \hamfree w_u \ri)_{\ldf(\R^6)} = F[w_v,w_u]<\infty, $$
					we must have 
					\bdm
						\int_{\R^6} \diff \kv_1 \diff \kv_2 \frac{\lf( \hat\xi_v^*(\kv_1) - \hat\xi_v^*(\kv_2) \ri) \lf(\hamfree w_u\ri)(\kv_1,\kv_2)}{ k_1^2 + k_2^2 + \frac{2}{m+1} \kv_1 \cdot \kv_2} < +\infty.
					\edm 
					On the other hand $ \hamfree w_u \in L^2(\R^6)  $ for any $ w_u $ and therefore $ \pot \xi_v \in L^2(\R^6) $. Hence we can write
					 \bdm
					 	\lf(w_v + \pot \xi_v, \hamfree w_u \ri)_{\ldf(\R^6)} =  \lf(w_v, \hamfree w_u \ri)_{\ldf(\R^6)}  + \lf(\pot \xi_v, \hamfree w_u \ri)_{\ldf(\R^6)}.
					\edm
					Since both terms on the r.h.s. are separately finite, we arrive at
					\bdm
					 2\qformo[\xi_v,\xi_u]= \lf(\pot \xi_v, \hamfree w_u \ri)_{\ldf(\R^6)}
					\edm
					The previous equation can be rewritten as:
					\beq
						\label{qformol identity}
						\qformo[\xi_v,\xi_u] = \lf(\xi_v, \lf. w_u\ri|_{\yv_2 = 0} \ri)_{L^2(\R^3)}.
					\eeq
					the r.h.s. of \eqref{qformol identity} is finite under our hypothesis. This statement will be proved in \eqref{w regularity} contained in Lemma \ref{lem: w regularity}.

Before going on we have to discuss the regularity of $\xi_u$ according to $q_u$. If $ q_u = 0 $, then $ \eta_u = \xi_u  $ and $\xi_u \in H^{1/2}(\R^3)$.
 Since $\qformo$ is closed and positive, identity \eqref{qformol identity} together with \eqref{w regularity} imply that $ \xi_u \in \dom_0 $ and
					\bdm
						\qformo[\xi_v,\xi_u] = \lf( \xi_v, \gammaol \xi_u \ri)_{L^2(\R^3)},
					\edm
					where with a small abuse of notation we have denoted the restriction of $\Gamma_0$ to  $\hilb_{1,n}$ by the same symbol, so that we obtain the boundary condition
					\beq
						\label{w boundary condition}
						w_u(\yv,0) = \lf( \Gamma_0 \xi_u \ri)(\yv).
					\eeq
					We remark that this boundary condition implies more regularity on $ \xi_u $ than the one assumed a priori in the domain definition \eqref{operator domain}: next Lemma \ref{lem: w regularity} shows that a consequence of \eqref{w boundary condition} 
					 is 
					\beq
						\label{more regularity}
						\Gamma_0 \xi_u \in H^{1/2}(\R^3).
					\eeq

On the opposite if $ q_u \neq 0 $, we should distinguish two cases according to $s$.
\begin{itemize}
	\item If $\mathit{1/2 < s(m) < 1} $ then	
						the same argument as before applies since by definition $ \xi_u \in H^{1/2}(\R^3) $;
	\item If $\mathit{0 < s(m) < 1/2} $ we can decompose 
						\beq
							\xi_u = \xi_< - \Xi_<,
						\eeq
						where $ \xi_< \in H^{1/2}(\R^3) $ and $\Xi_<\in \dot{H}^{1/2}(\R^3) $ are given by (recall the notation \eqref{notation abuse 1})
						\beq \label{fine1}
							\hat\xi_< : = \hat\xi_u + q_u \one_{\{ k \leq \eps \}}  \widehat{\Xi^-}(\kv),
						\eeq
						and
						\beq \label{fine2}
							\widehat{\Xi_<}(\kv) : = q_{u} \one_{\{ k \leq \eps \}} \widehat{\Xi^-}(\kv).
						\eeq
						Then the condition \eqref{qformol identity} becomes
						\bdm
							\qformo[\xi_v, \xi_<] = \lf(\xi_v, \lf. w_u\ri|_{\yv_2 = 0} \ri)_{L^2(\R^3)} + \qformo[\xi_v, \Xi_<],
						\edm
						where all the three terms are separately finite. Now we notice that although $ \Xi_< \notin L^2(\R^3) $ we have
							$\Gamma_0 \Xi_< \in L^2(\R^3) $
						where $ \Gamma_0 $ stands now for the formal action of the integral operator \eqref{Gamma0} due to the cut-off we have introduced.
 
						\bdm
							\qformo[\xi_v, \Xi_<] = \lf( \xi_v, \gammaol \Xi_< \ri)_{L^2(\R^3)},
						\edm
						and we obtain
						\beq
							\qformo[\xi_v, \xi_<] =  \lf(\xi_v, \lf. w_u\ri|_{\yv_2 = 0} + \gammaol \Xi_<  \ri)_{L^2(\R^3)}.
						\eeq
						Since $ \lf. w_u\ri|_{\yv_2 = 0} + \gammaol \Xi_< \in L^2(\R^3) $ by \eqref{w regularity}, this implies that $ \xi_< $ must belong to $ \dom_0 $ 
and the boundary condition \eqref{w boundary condition} is satisfied as well. Notice however that since in this case $ \xi_u \notin L^2(\R^3) $ the expression $ \Gamma_0 \xi_u $ should be meant as the formal action of the operator, i.e., \eqref{Gamma0}.
	\end{itemize}				
					
			\item 	To complete the derivation of $ \dom $, it remains to verify what is the condition implied by a more generic $ v $, decomposing as $ v = w_v + \pot(\xi_v + q_v \xximinus) $: in this case we get
					\bml{
						\label{identity dom 3}
						\formol[v,u] = F[w_v,w_u] + 2\qformo[\xi_v,\xi_u] + \beta q_v^* q_u = \lf(v, \psi\ri)_{\ldf(\R^6)} 	\\
						= \lf(w_v, \hamfree w_u \ri)_{\ldf(\R^6)} + \lf( \pot \lf(\xi_v + q_v \xximinus\ri), \hamfree w_u \ri)_{\ldf(\R^6)}.
					}
					After a cancellation, using \eqref{w boundary condition} we arrive at
					\beq
						\label{identity dom 4}
						2\qformo[\xi_v,\xi_u] + \beta q_v^* q_u 
= 2 \lf(  \xi_v + q_v \xximinus,  \Gamma_0 \xi_u \ri)_{L^2(\R^3)}.
					\eeq
					Now we would like to decompose the scalar product into the sum of two scalar products in order to exploit the cancellation with the term $ 2 \qformo $ on the l.h.s., but in order to do that we have to make a distinction according to the value of $ s(m) $:
		\begin{itemize}		
			\item	If $\mathit{1/2 < s(m) < 1} $ since $ \xi_v \in H^{1/2}(\R^3) $ we can break the scalar product obtaining
							\beq
								\beta q_u = 2 \lf( \xximinus, \Gamma_0 \xi_u \ri)_{L^2(\R^3)},
							\eeq
							i.e.,
\[
\lim_{\varepsilon \to 0}  \lf( \xximinus, \Gamma_0 \xi_u \ri)_{L^2(\R^3_{\varepsilon})} =      \lf( \xximinus, \Gamma_0 \xi_u \ri)_{L^2(\R^3)}
\]
by dominated convergence and \eqref{operator domain} is proven.

			\item		If	$\mathit{0 < s(m) < 1/2} $, we have to go through a limit procedure: for any $ \eps > 0 $, we have
					\bmln{
						\lf(  \xi_v + q_v \Xi^-,  \Gamma_0 \xi_u \ri)_{L^2(\R^3)} = \lf(  \xi_<,  \Gamma_0 \xi_u \ri)_{L^2(\R^3)} + \lf( \Xi_>,  \Gamma_0 \xi_u \ri)_{L^2(\R^3)} 
						\\
						= \qformo[\xi_<,\xi_u] + \lf(\Xi_>, \Gamma_0   \xi_u \ri)_{L^2(\R^3)},
					}
					where 
					\beq
						\Xi_> : = \xximinus - \Xi_<
					\eeq
The l.h.s. is finite therefore in order to prove that 
this decomposition is meaningful it is sufficient to prove that $\lf(\Xi_>, \Gamma_0   \xi_u \ri)_{L^2(\R^3)}$ is finite. By  \eqref{more regularity}, we have
					\bmln{
						\lf| \lf(\Xi_>, \Gamma_0   \xi_u \ri)_{L^2(\R^3)} \ri| \leq \lf\|   \Xi_> \ri\|_{H^{-1/2}(\R^3)}  \lf\| \Gamma_0 \xi_u \ri\|_{H^{1/2}(\R^3)}	\\
						\leq c
 \lf\| \Gamma_0 \xi_u \ri\|_{H^{1/2}(\R^3)} \bigg( \int_{\eps}^{\infty} \diff k \: k^{-3+ 2s(m)} \bigg)^{1/2} < \infty.
					}
					Due to the definition of $ \Xi_> $ in \eqref{identity dom 4}, we get
					\beq
						\beta q_v^* q_u = 2 \qformo[\Xi_<,\xi_u] + 2  q_{v}^* \lf(\ximinus, \Gamma_0   \xi_u \ri)_{L^2(\R_{\eps}^3)}.
					\eeq
					Notice that as $ \eps \to 0 $ the first term on the r.h.s. vanishes:
					\bdm
						\lf| \qformo[\Xi_<,\xi_u] \ri| \leq C \lf\| \xi_u \ri\|_{\dot{H}^{1/2}(\R^3)} \int_0^{\eps} \diff k \: k^{-1+2s(m)} \leq C \eps^{2s(m)} \underset{\eps \to 0}{\longrightarrow} 0,
					\edm
					since $ \xi_u \in \dot{H}^{1/2}(\R^3) $ by hypothesis and \eqref{operator domain} is proven.
	\end{itemize}

		\end{enumerate}
	\end{proof}
	
	\begin{lem}
		\label{lem: w regularity}
		\mbox{}	\\
		Let $ u \in \mathscr{D}\lf(\hamol\ri) $ and $ w_u $ be its regular part of $ u $, then 
		\beq
			\label{w regularity}
			w_u(\yv, 0) \in H^{1/2}(\R^3).
		\eeq
	\end{lem}	
	
	\begin{proof}
		We first note that $ u \in \mathscr{D}\lf(\hamol\ri) $ implies the decomposition $ u = w_u + \pot \eta $ with $ w_u \in \dot{H}^2_{\mathrm{f}}(\R^6) $. However by decomposing $ \pot \eta $ into high and low frequency contributions one immediately obtains that  $ w_u + \pot_< \eta \in L^2_{\mathrm{f}} (\R^6) $ , and therefore  also $ w_u + \pot_< \eta \in \hdf(\R^6) $ holds,
where $ \pot_< \eta $ stands for the potential of the charge $ \eta $ cut for momenta 
$ k_1^2 + k_2^2 \leq \eps^2 $, $ \eps > 0 $, i.e., 
		\bdm
			\widehat{\pot_< \eta}(\kv_1,\kv_2)  : = \one_{\{ k_1^2 + k_2^2 \leq \eps^2 \}} \frac{\hat{\eta}(\kv_1) - \hat{\eta}(\kv_2)}{k_1^2 + k_2^2 + \frac{2}{m+1} \kv_1 \cdot \kv_2}. 
		\edm
		Then by standard Sobolev trace theorems, the trace $ (w_u+ \pot_< \eta)(\yv,0) $ belongs to $ H^{1/2}(\R^3) $. To complete the proof it only remains to show that $ (\pot_< \eta)(\yv,0) $ belongs to $ H^{1/2}(\R^3) $. This can be proven by direct inspection thanks to the assumption $ \eta \in L^2(\R^3) $.
In order to show it, we write
		\bdm
			g(\yv) : = \lf( \pot_< \eta \ri) (\yv,0) = \frac{1}{(2\pi)^3} \int_{k_1^2 + k_2^2 \leq \eps^2} \diff \kv_1 \diff \kv_2 \: \frac{e^{i \kv_1 \yv} \lf( \hat{\eta}(\kv_1) - \hat{\eta}(\kv_2) \ri)}{k_1^2 + k_2^2 + \frac{2}{m+1} \kv_1 \cdot \kv_2},
		\edm
		and compute
		\bdm
			\hat{g}(\kv) = \int_{\R^3} \diff \kv_2 \: \frac{\one_{\{ k^2 + k_2^2 \leq \eps^2 \}} \lf( \hat{\eta}(\kv) - \hat{\eta}(\kv_2) \ri)}{k^2 + k_2^2 + \frac{2}{m+1} \kv \cdot \kv_2},
		\edm
		so that
		\bmln{
			\lf| \hat{g}(\kv) \ri| \leq C \one_{\{ k \leq \eps \}} \int_{k_2 \leq \eps} \diff \kv_2 \: \frac{ \lf| \hat{\eta}(\kv) - \hat{\eta}(\kv_2) \ri|}{k^2 + k_2^2} 	\\
			\leq C \one_{\{ k \leq \eps \}} \bigg[ \lf| \hat{\eta}(\kv) \ri| + \lf\| \eta \ri\|_{H^{-1/2}(\R^3)} \bigg( \int_0^{\eps} \diff k_2 \frac{k_2^2 \sqrt{1 + \eps^2}}{(k^2 + k^2_2)^2} \bigg)^{1/2} \bigg],	\\
			\leq C \one_{\{ k \leq \eps \}} \lf[ \lf| \hat{\eta}(\kv) \ri| + \lf\| \eta \ri\|_{H^{-1/2}(\R^3)} k^{-1} \ri],
		}
		which compactly supported and square integrable as $ k \to 0 $, thanks to the hypothesis on $ \eta $, i.e., $ \eta \in H^{-1/2}(\R^3) $, 
then $g$ is square integrable. Since $\hat g$ is compactly supported then $g\in H^p (\R^3) $ for $p\geq 0$
	\end{proof}
	
		We now prove Proposition \ref{pro: STM} and, as in the above analysis, we restrict ourselves to the operator $ \hamol $.
		
		\begin{proof}[Proof of Proposition \ref{pro: STM}]
		Let $ u \in \dom(\hamol) $ with $ \xi_u \in C^{\infty}_0(\R^3) \cap \hilb_{1,n} $, then we want to show that 
		\bdn
			\hamol u & = &\hamfree w_u,	\nonumber		\\
			w_u(\yv, 0) & = & \lf( \gammaol \eta_u \ri)(\yv),
		\edn
		which is equivalent to prove that 
		\beq
			q_u = 0.
		\eeq
		Notice that we have taken a smooth charge $ \xi_u $ as in the definition of $ \tilde{H}_0 $ but a direct inspection of the following argument reveals that it could be applied as well to any $ \xi_u $, e.g., in $ H^{3/2}(\R^3) $.
		
		The only non trivial point is the analysis of the second boundary condition, which must result in a trivial identity. Indeed we shall now prove that, for any such $ \xi_u $, one has
		\bdm
			\lim_{\eps \to 0} \lf( \xximinus, \Gamma_0 \xi_u \ri)_{L^2(\R_{\eps}^3)} = 0.
		\edm
		The simplest way to prove it is by noting that the order of the integrals in the above expression can be exchanged: 
		\bmln{
			\lim_{\eps \to 0} \lf( \xximinus, \Gamma_0 \xi_u \ri)_{L^2(\R^3_{\eps})} = \lim_{\eps \to 0} \int_{\R^3_{\eps}} \diff \pv \: \xximinus(\pv)  \bigg\{ 2\pi^2  \tx\sqrt{\frac{m(m+2)}{(m+1)^2}} p \: \widehat{\xi_u}(\pv) + \disp\int_{\R^3_{\eps}} \diff \qv  \: \frac{\widehat{\xi_u}(\qv)}{p^2 + q^2 + \frac{2}{m+1} \pv \cdot \qv} \bigg\}		\\
			=	 \lim_{\eps \to 0} \bigg[2\pi^2  \tx\sqrt{\frac{m(m+2)}{(m+1)^2}} \disp\int_{\R_{\eps}^3} \diff \pv \: p \: \xximinus(\pv)  \widehat{\xi_u}(\pv) + \int_{\R^3_{\eps}} \diff \pv \diff \qv  \: \frac{\xximinus(\pv) \widehat{\xi_u}(\qv)}{p^2 + q^2 + \frac{2}{m+1} \pv \cdot \qv} \bigg]	\\
			= \lim_{\eps \to 0} \lf( \Gamma_0 \xximinus, \xi_u \ri)_{L^2(\R^3_{\eps})} = 0,
		}
		where we have exchanged the order of integration thanks to the finiteness of both integrals, as one can easily prove by exploiting the regularity of $ \xi_u $ and Cauchy-Schwarz inequality:
		\bmln{
			\bigg| \disp\int_{\R_{\eps}^3} \diff \pv \: p \:\xximinus(\pv) \widehat{\xi_u}(\pv) \bigg| \leq C \bigg[ \bigg( \int_{\R^3} \diff \pv \: p^3 \lf| \widehat{\xi_u}(\pv) \ri|^2 \bigg)^{1/2} \bigg( \int_{1}^{\infty} \diff p \ p^{-3+2s(m)} \bigg)^{1/2} 	\\
			+ \bigg( \int_{0}^1 \diff p \: p^{2 s(m)} \bigg)^{1/2} \bigg( \int_{p \leq 1} \diff \pv \:  \lf| \widehat{\xi_u}(\pv) \ri|^2 \bigg)^{1/2} \bigg] \leq C \lf\| \xi_u \ri\|_{H^1(\R^3)},
		}
		for any $ 0 < s(m) < 1 $. As last thing, we prove that $\Gamma_0 \xximinus=0$. For simplicity  from now on we make a little abuse of notation, setting
		\beq
			\label{notation abuse 2}
			\Gamma_0 : = \lf. \Gamma_0 \ri|_{\hilb_{1,n}},
		\eeq
		and compute
		\bml{
			\label{gamma0 vanishes}
			\lf(\Gamma_0 \ximinus\ri)(p)  = 2\pi^2 \sqrt{\frac{m(m+2)}{(m+1)^2}} \f{1}{p^{1-s} }+2\pi  \int_{-1}^1 \diff t \: t \int_0^{\infty} \diff q \: \frac{q^{s}}{q^2 + p^2 + \frac{2}{m+1} t pq} \\
			= \f{2\pi}{p^{1-s}} \lf[ \pi \sqrt{\frac{m(m+2)}{(m+1)^2}} +  \int_{-1}^1 \diff t \: t \int_0^{\infty} \diff q \: \frac{q^{s}}{q^2 + 1 + \frac{2}{m+1} t q}\ri] =0,
		}
	where in the last step we have used the definition of $s(m)$ (see \eqref{exponent s}).
	\end{proof}
	
	\n
	We conclude the Section with an investigation of the asymptotic behavior of charges in $ \dom(\hamol) $.
	
	\begin{proof}[Proof of Proposition \ref{pro: charge asympt}] Let $ \eta $ be an admissible charge belonging to $ \dom(\hamol) $ and $ \xi $ its regular part, then we can find $ \xi $ of the form
	\beq
		\hat{\xi}_{>}(\kv) = \lf( \frac{A}{k^{2+s(m)}} \one_{\{ k \geq R \}} + \chi(k) \ri) Y_{1}^{n}(\vartheta_k, \varphi_k),
	\eeq
	with $ \chi \in H^{1/2}(\R^3) $ and
	\beq
		\label{assumption chi}
		\lf| \chi(k) \ri| \leq \frac{C}{k^{\gamma}},	\qquad		\gamma > 2 + s(m),
	\eeq
	for $0<s(m)<1/2$. The same argument can be repeated when $ 1/2 < s(m) < 1 $ taking the whole $\xi $ as above.
	For such regular parts the overall charge $ \eta  $ is
	\bml{
		\hat{\eta}(\kv) = \lf( \frac{q}{k^{2- s(m)}} + \frac{A}{k^{2+s(m)}} \one_{\{ k \geq R \}} + \chi(k) \ri) Y_{1}^{n}(\vartheta_k, \varphi_k):=	\\
\lf( q \ximinus (k) + A \xiplus(k) \one_{\{ k \geq R \}} + \chi(k) \ri) Y_{1}^{n}(\vartheta_k, \varphi_k)
	}	
	where 
	\beq
		\xiplus(k) : = \frac{1}{k^{2 + s(m)}}
	\eeq
	and the coefficient $ A $ must satisfy the second boundary condition, i.e., 
	\bml{
		\label{STM 1}
		\beta q = \lim_{\eps \to 0} \lf( \ximinus, \Gamma_0 \xi \ri)_{L^2(\R_{\eps}^3)} = \lim_{\eps \to 0} \lf[ A \lf( \ximinus, \Gamma_0 \lf( \one_{\{ k \geq R \}} \xiplus \ri) \ri)_{L^2(\R^3_{\eps})} + \lf( \ximinus, \Gamma_0 \chi \ri)_{L^2(\R^3_{\eps})}  \ri] \\
		= A \lim_{\eps \to 0} \lf( \ximinus, \Gamma_0 \lf( \one_{\{ k \geq R \}} \xiplus \ri) \ri)_{L^2(\R^3_{\eps})}.
	}
	The vanishing of the second term can be proven as in the proof of Proposition \ref{pro: STM}. Since the diagonal term in the expression is absolutely convergent
	\beq
		\bigg| \int_{\R_{\eps}^3} \diff \kv \: \frac{1}{k^{1- s(m)}} \chi(k) \bigg| \leq C \int_{\eps}^{\infty} \diff k \: k^{-\gamma+1 + s(m)} < + \infty,
	\eeq
	thanks to \eqref{assumption chi}, then the order of the integrals in the off-diagonal term can be exchanged and $ \Gamma_0  \ximinus=0 $ appears.
	
	Then it remains to compute the integral in \eqref{STM 1}. Using \eqref{exponent s} which provides cancellations and scale invariance w.r.t. $R$, we have
	\bml{
		\lim_{\eps \to 0} \lf( \ximinus, \Gamma_0 \lf( \one_{\{ k \geq R \}} \xiplus \ri) \ri)_{L^2(\R^3_{\eps})} =  \disp\int_{R}^{\infty} \diff p \: \frac{1}{p} \bigg[ 2\pi^2 \tx\sqrt{\frac{m(m+2)}{(m+1)^2}} + 2\pi \disp\int_{-1}^1 \diff t \: t \int_{R/p}^{\infty} \diff q  \: \frac{q^{-s(m)}}{q^2 +1 + \frac{2}{m+1} q t} \bigg]	\\
		=  -2\pi \disp\int_{1}^{\infty} \diff p \: \frac{1}{p}\disp\int_{-1}^1 \diff t \: t \int^{1/p}_{0} \diff q  \: \frac{q^{-s(m)}}{q^2 +1 + \frac{2}{m+1} q t} = \nu(m).
	}
	In order to see that $ \nu(m) $ is finite it suffices to estimate the large $ p $ asymptotics of the last integral, which behaves like $ p^{-1+s(m)} $, so that the integrand goes asymptotically as $ p^{-2+s(m)} $, which is integrable. Moreover splitting the integral over $t$ one gets
	\beq
		\nu(m) =  \frac{8 \pi}{m+1} \disp\int_{1}^{\infty} \diff p \: \frac{1}{p}\disp\int_{0}^1 \diff t \: t^2 \int^{1/p}_{0} \diff q \frac{q^{1-s(m)}}{(q^2 +1)^2 - \frac{4}{(m+1)^2} q^2},	
	\eeq
	which shows that $ \nu(m) > 0 $. 
	
	Then going back to \eqref{STM 1}	
	\beq
		A = \frac{\beta q}{\nu(m)},		
	\eeq
	and asymptotically
	\beq
		\eta_{1}(k) \propto \frac{q}{k^{2 - s(m)}} + \frac{\beta q}{\nu(m)}	\frac{1}{ k^{2 + s(m)}} + o(k^{-2-s(m)}).
	\eeq
	\end{proof}




\vs

\appendix

\section{Critical Masses}
\label{sec: critical masses}
In this Appendix we study equation \eqref{exponent s} and we prove that the critical masses $m^\star $, $m ^{\star \star}$,  are well defined and $m ^{\star}<m ^{\star \star}$. Some of these properties were proved in \cite{CDFMT}
and \cite{FT} in a lesser generality but we prefer to give the analysis here for sake of completeness.

Let us denote by $F_{\ell} (m,\, s): [0,+\infty) \times [0,1]\to \R$ the following function:
\beq
		F_{\ell} (m,\, s) :=
		\pi \sqrt{\frac{m(m+2)}{(m+1)^2}} +  \int_{-1}^1 \diff t \: P_{\ell}(t) \int_0^{\infty} \diff p \: \frac{p^{s}}{p^2 + 1 + \frac{2}{m+1} t p} =: F_{\ell,1} (m)+ F_{\ell,2} (m,\, s) .
	\eeq
\begin{lemma}[Properties of $ F_{\ell} $]
	\label{lemmafunzione}
	\mbox{}	\\
For odd $\ell$ the function  $F_{\ell} $ enjoys the following properties:
\begin{enumerate}[a)]
\item $F_{\ell}$ is continuous and bounded;
\item for fixed $s$,  $F_{\ell} (\cdot ,\, s)$ is an increasing function; moreover  $F_{\ell} (0,\, s)<0$ and  $\lim_{m\to \infty} F_{\ell} (m,\, s)=2\pi^2$;
\item  for fixed $m$,  $F_{\ell} (m, \cdot)$ is a decreasing function;
\item for fixed $(m,s)$, we have $F_{\ell} (m ,s)< F_{\ell+2} (m ,s) $.
\end{enumerate}
\end{lemma}
\begin{proof}
Continuity in the interior of the domain of definition follows from the integrability of the integrand. 
Some remarks are in order when
$s\to1 $ or $m\to 0$. Using the parity of Legendre polynomials we can cast  $F_{\ell}$ in the following form:
\begin{align}
F_{\ell} (m,\, s) &=
		\pi \sqrt{\frac{m(m+2)}{(m+1)^2}} 
+  \int_{0}^1 \diff t \: P_{\ell}(t) \int_0^{\infty} \diff p \: p^{s}\left( \frac{1}{p^2 + 1 + \frac{2}{m+1} t p}- \frac{1}{p^2 + 1 - \frac{2}{m+1} t p} \ri) \nonumber  \\
 &= \pi \sqrt{\frac{m(m+2)}{(m+1)^2}} 
-   \frac{4}{m+1} \int_{0}^1 \diff t \: t\, P_{\ell}(t) \int_0^{\infty} \diff p \: 
\f{ p^{s+1}}{( p^2 + 1 + \frac{2}{m+1} t p)( p^2 + 1 - \frac{2}{m+1} t p)}, \label{equivalent}
\end{align}
and we see that the limit $s \to 1 $ is harmless by dominated convergence. When $m \to 0$ a summable singularity appears in the integral
with respect to $t$, that is as $t \to -1$ we have
\[
 \int_0^{\infty} \diff p \: \frac{p^{s}}{p^2 + 1 + 2 t p} \sim \f{1}{\sqrt{t+1} } .
\]
Due  to the integrability of the singularity,  $F_{\ell} $ is finite and continuous by dominated convergence and property $a)$ is proved.

In order to prove $b)$, it is convenient to use a series representation of  $F_{\ell} $, that is:
\begin{align*}
F_{\ell} (m,\, s)& =
	\pi \sqrt{1-\frac{1}{(m+1)^2}} 
+ \int_{-1}^1 \diff t \: P_{\ell}(t) \int_0^{\infty} \diff p \: \frac{p^{s}}{p^2 + 1 }  \sum_{k=0}^\infty \lf( - \frac{2}{m+1} \f{t p}{p^2+1}\ri)^k \\
& = \pi \sqrt{1-\frac{1}{(m+1)^2}} 
+   \sum_{k=0}^\infty \lf( \f{-2}{m+1}\ri)^k   \int_{-1}^1 \diff t \: t^k \, P_{\ell}(t) 
\int_0^{\infty} \diff p \: \frac{p^{s+k}}{(p^2 + 1)^{k+1} }.
\end{align*}  
Notice that for even $k$
\[
 \int_{-1}^1 \diff t \: t^k \, P_{\ell}(t) =0
\]
and therefore only the odd terms actually appear in the series. In facts, we have
\beq \label{series1}
F_{\ell} (m,\, s)= 
 \pi \sqrt{1-\frac{1}{(m+1)^2}} 
-   \sum_{n=0}^\infty \lf( \f{2}{m+1}\ri)^{2n+1}   \int_{-1}^1 \diff t \: t^{2n+1} \, P_{\ell}(t) 
\int_0^{\infty} \diff p \: \frac{p^{s+2n+1}}{(p^2 + 1)^{2n+2} }.
\eeq
Moreover, using the definition of Legendre polynomial
\[
P_{\ell}(t) = \frac{1}{2^{\ell} \ell ! } \f{d^l}{dt^l } (t^2-1)^{\ell} ,
\]
and integrating by parts, we have
\[
 \int_{-1}^1 \diff t \: t^{2n+1} \, P_{\ell}(t)  = 
\begin{cases}
0, & \text{ if } 2n+1<\ell, \\
\displaystyle{\f{1}{2^l } \binom{2n+ 1}{\ell} \int_{-1}^1 \diff t \: t^{2n+1-\ell } \, (1-t^2)^{\ell}},  &  \text{ if } 2n+1\geqslant \ell.
\end{cases}
\]
With the change of index $2n+1-\ell=2k$ in the series, \eqref{series1} becomes
\bml{ \label{series2}
F_{\ell} (m,\, s)= 
 \pi \sqrt{1-\frac{1}{(m+1)^2}} +\\
- \f{1}{2^{\ell} }  \sum_{k=0}^\infty  \lf( \f{2}{m+1}\ri)^{2k+\ell} \binom{2k+\ell}{\ell}  \int_{-1}^1 \diff t \: \, (1-t^2)^{\ell}\,  t^{2k}
\int_0^{\infty} \diff p \: \frac{p^{s+2k+\ell}}{(p^2 + 1)^{2k+\ell+1} }.
}
Since in  \eqref{series2} it appears, up to a global sign,  an absolutely convergent, positive term series  of monotone functions, 
then $ F_{\ell,2}$ is negative and increasing.
Since $ F_{\ell,1}$ is increasing the also $ F_{\ell}$ is an increasing function of $m$ for fixed $s$. 
Clearly $ F_{\ell,1}(0)=$ and $ F_{\ell,2}(0,s)<0$ due to the above representation. 
Notice that the finiteness of the series is not obvious but we know that this is the case due to the integral representation and the above remarks.
Concerning the behavior at infinity, by representation  \eqref{equivalent} we have
\[
 F_{\ell,2}(m,s) =
-   \frac{4}{m+1} \int_{0}^1 \diff t \: t\, P_{\ell}(t) \int_0^{\infty} \diff p \: 
\f{ p^{s+1}}{( p^2 + 1 + \frac{2}{m+1} t p)( p^2 + 1 - \frac{2}{m+1} t p)} .
\]
Since
\[
\f{1}{ p^2 + 1 - \frac{2}{m+1}  p} \leq \f{m+1}{m} \f{1}{ p^2 + 1 },
\]
we have
\bmln{
 |F_{\ell,2}(m,s) |  \leq
  \frac{4}{m+1}\lf| \int_{0}^1 \diff t \: t\, P_{\ell}(t)  \ri| \int_0^{\infty} \diff p \: 
\f{ p^{s+1}}{( p^2 + 1 )( p^2 + 1 - \frac{2}{m+1}  p)}  	\\
\leq
  \frac{4}{m}\lf| \int_{0}^1 \diff t \: t\, P_{\ell}(t)  \ri| \int_0^{\infty} \diff p \: 
\f{ p^{s+1}}{( p^2 + 1 )^2},
}
then
\[
\lim_{m\to \infty}  F_{\ell,2}(m,s) =0
\]
and $b)$ is proved. In order to prove property $c)$, it is sufficient to prove that each term
of the representation of $F_{\ell,2}$ by series 
 in \eqref{series2} is a monotone function of $s$ due to the positivity of coefficients.
This is straightforward since
\bmln{
\int_0^{\infty} \diff p \: \frac{p^{s+2k+\ell}  }{ (p^2 + 1)^{2k+\ell+1} } =
\int_{\R} \diff x\: \frac{e^{ (2k+\ell+1)x }}{(e^{2x} + 1)^{2k+\ell+1} } e^{sx} 
=\int_{\R} \diff x\: \frac{ e^{sx} }{(2 \cosh x )^{2k+\ell+1} }
\\
= 2\int_0^\infty \diff x\: \frac{ \sinh (sx) }{(2 \cosh x )^{2k+\ell+1} }
}
and $c)$ follows from the monotonicity of $\sinh$. Last we prove $d)$. Since only  $F_{\ell,2}$ actually depends on $\ell$ ,
we have to prove  $F_{\ell,2} (m, \cdot ,s)< F_{\ell+2,2} (m, \cdot ,s) $.
We rewrite  $F_{\ell,2} (m, \cdot ,s)$ extracting the first term of the series.
\bml{ \label{first}
F_{\ell,2} (m, \cdot ,s) =
- \f{1}{2^{\ell} }   \lf( \f{2}{m+1}\ri)^{\ell} \int_{-1}^1 \diff t \: \, (1-t^2)^{\ell}
\int_0^{\infty} \diff p \: \frac{p^{s+\ell}}{(p^2 + 1)^{\ell+1} } \\
 - \f{1}{2^{\ell} }  \sum_{k=0}^\infty  \lf( \f{2}{m+1}\ri)^{2k+2+\ell} \binom{2k+2+\ell}{\ell}  \int_{-1}^1 \diff t \: \, (1-t^2)^{\ell}\,  t^{2k+2}
\int_0^{\infty} \diff p \: \frac{p^{s+2k+2+\ell}}{(p^2 + 1)^{2k+2+\ell+1} }.
}
This has to be compared with:
\bml{ \label{second}
F_{\ell+2,2} (m, \cdot ,s) =  - \f{1}{2^{\ell+2} }  \sum_{k=0}^\infty  \lf( \f{2}{m+1}\ri)^{2k+2+\ell} \binom{2k+2+\ell}{\ell+2}  \times	\\
\times \int_{-1}^1 \diff t \: \, (1-t^2)^{\ell+2}\,  t^{2k}
\int_0^{\infty} \diff p \: \frac{p^{s+2k+2+\ell}}{(p^2 + 1)^{2k+2+\ell+1} }.
}
A comparison between \eqref{first} and \eqref{second} shows that $d)$ holds true if 
\[
 \f{1}{2^{\ell+2} }  \binom{2k+2+\ell}{\ell+2}  \int_{-1}^1 \diff t \: \, (1-t^2)^{\ell+2}\,  t^{2k} <
 \f{1}{2^{\ell} }  \binom{2k+2+\ell}{\ell}  \int_{-1}^1 \diff t \: \, (1-t^2)^{\ell}\,  t^{2k+2},
\]
that is 
\beq \label{third}
\f{1}{4} (2k+2)(2k+1) \int_{-1}^1 \diff t \: \, (1-t^2)^{\ell+2}\,  t^{2k} < 
 (\ell +2)(\ell +1)\int_{-1}^1 \diff t \: \, (1-t^2)^{\ell}\,  t^{2k+2}.
\eeq
We rewrite the l.h.s. of \eqref{third} integrating by parts twice
\bmln{
\f{ (2k+2)(2k+1)}{4} \int_{-1}^1 \diff t \: \, (1-t^2)^{\ell+2}\,  t^{2k}  = \f{1}{4} \int_{-1}^1 \diff t \: \, (1-t^2)^{\ell+2}\,  \f{d^2 t^{2k+2} }{dt^2}	\\
= \f{1}{4} \int_{-1}^1 \diff t \: \, \f{d^2}{dt^2}  (1-t^2)^{\ell+2} \; t^{2k+2} \\
= (\ell +2)(\ell +1)\int_{-1}^1 \diff t \: \, (1-t^2)^{\ell}\,  t^{2k+4 }
 -\f{\ell+2}{2} \int_{-1}^1 \diff t \: \, (1-t^2)^{\ell+1}\,  t^{2k+2 } \\
< (\ell +2)(\ell +1)\int_{-1}^1 \diff t \: \, (1-t^2)^{\ell}\,  t^{2k+2 }
 -\f{\ell+2}{2} \int_{-1}^1 \diff t \: \, (1-t^2)^{\ell+1}\,  t^{2k+2 },
}
and $d)$ is proved. 
\end{proof} 

\begin{proposition}[Critical masses]
	\label{pro: critical masses}
	\mbox{}	\\
Equation  \eqref{exponent s} defines a continuous increasing function $m(s)$.
	Moreover the critical masses $m^\star$, $m^{\star \star}$, are well defined by \eqref{critical masses}
and they satisfy  $ \mstar < \mss  $.
\end{proposition}
\begin{proof}
For fixed odd $\ell$, let us consider equation
\beq \label{dini}
F_{\ell} (m,s)=0.
\eeq
By Dini's theorem, and due to $a)$ and $b)$ of Lemma \ref{lemmafunzione}, 
equation \eqref{dini} defines a continuous function $m_{\ell} (s): [0,1]\to [0,+\infty)$. Moreover it is monotone increasing in $s$
by $c)$. Notice also that  the functions   $m_{\ell} (s)$ are decreasing in  $\ell$ due to $d)$.
We set  $m(s) : = m_1 (s)$ and 
\[
 \mstar : =  m (0),  \qquad \mss : =  m (1)
\]
and $ \mstar < \mss $ follows from the monotonicity of $m (s)$.

\end{proof}

\noindent
Notice that if $\ell$ is even then $F_{\ell , 2} >0$ and \eqref{dini} has no solutions.
Furthermore one could  try repeating the arguments in this paper for $\ell>1$, introducing
$s_{\ell} (m)$ as the inverse of and  $m_{\ell} (s)$ and 
\[
\widehat{\Xi_{\ell,n}^-} = \ximinus_{\ell}  Y_{\ell}^n \,, \;\;\;\;\;\;\;\; \ximinus_{\ell} (k)= \f{1}{k^{2-s_{\ell} (m)}}
\]
which are the cornerstones of our construction. The proofs would work without modifications and we
could construct extension of the STM operators with singular charge asymptotic for several $\ell$.
In order to be able to do so, it is necessary that  the intervals $[m^{\star}_{\ell}, m^{\star\star}_\ell]$
overlap for different $\ell$, where we have put 
\beq
	\label{m3star}
\mstar_{\ell} :=  m_{\ell} (0),  \qquad \mss_{\ell} : = m_{\ell} (1).
\eeq
A numerical analysis shows that 
\[
 \mss_{3} = 0.0142.
\]
Since $ \mss_{3} < m^\star$  no overlap is possible due to monotonicity properties already proved.

\section{Admissible Charges and Form Domain}
\label{charge}

In this Appendix we investigate the conditions to impose on the charge $ \eta $ so that it belongs to $ \dom[\formo] $. It is clear that in order for the decomposition \eqref{form domain} to make sense and to fulfill the request $ u \in \ldf(\R^6) $, it must be
\beq
	\pot \eta \in L^2_{\mathrm{loc}}(\R^6),
\eeq
where we have removed the fermionic restriction since it is satisfied by construction. In next Lemma we are going to show that this is equivalent to assume that
\beq
	\label{cond eta 1}
	\int_{\R^3} \diff \pv \: \one_{\{p \geq \eps\}} \: p^{-1} \lf| \hat{\eta}(\pv) \ri|^2 < \infty,
\eeq
for any $ \eps > 0 $. In other words we will prove that
\beq
	\pot \eta \in L^2_{\mathrm{loc}}(\R^6) \Longleftrightarrow \eta \in H^{-1/2}_{\mathrm{loc}}(\R^3).
\eeq

\begin{lem}
	\label{lem: pot norm 1}
	\mbox{}	\\
	For any $ \eps > 0 $ there exist two constants $ 0 < c_1, c_2 < \infty $, such that
	\bml{
		\label{pot norm 1}
		c_1 \int_{\R^3} \diff \pv \: \one_{\{p \geq \eps\}} \: p^{-1} \lf| \hat{\eta}(\pv) \ri|^2  \leq \int_{\R^6} \diff \kv_1 \diff \kv_2 \: \one_{\{k_1 \geq \eps\}} \one_{\{k_2 \geq \eps\}} \: \lf| \widehat{\pot \eta} (\kv_1,\kv_2) \ri|^2 	\\
		\leq	c_2 \int_{\R^3} \diff \pv \: \one_{\{p \geq \eps\}} \: p^{-1} \lf| \hat{\eta}(\pv) \ri|^2.
	}
\end{lem}

\begin{proof}
	By the definition \eqref{potential} we have  (we drop for simplicity the $ \hat{\mbox{}} $ denoting the Fourier transform)
	\beq
		\label{l2 norm pot}
		\int_{\R^6} \diff \kv_1 \diff \kv_2 \: \one_{\{k_1 \geq \eps\}} \one_{\{k_2 \geq \eps\}} \: \lf| \pot \eta (\kv_1,\kv_2) \ri|^2 = \int_{\R^6} \diff \kv_1 \diff \kv_2 \: \one_{\{k_1 \geq \eps\}} \one_{\{k_2 \geq \eps\}} \: \frac{\lf| \eta (\kv_1) - \eta(\kv_2) \ri|^2}{(k_1^2 + k_2^2 + \frac{2}{m+1} \kv_1 \cdot \kv_2)^2}.
	\eeq
The r.h.s. of the above expression is bounded by
	\bml{
		\label{pot norm ub}
		4 \int_{\R^6} \diff \kv_1 \diff \kv_2 \: \one_{\{k_1 \geq \eps\}} \one_{\{k_2 \geq \eps\}} \: \frac{\lf| \eta (\kv_1) \ri|^2}{(k_1^2 + k_2^2 + \frac{2}{m+1} \kv_1 \cdot \kv_2)^2} 	\\
		\leq 4 \int_{\R^6} \diff \kv_1 \diff \kv_2 \: \one_{\{k_1 \geq \eps\}} \one_{\{k_2 \geq \eps\}} \: \frac{\lf| \eta (\kv_1) \ri|^2}{(k_1^2 + k_2^2 - \frac{2}{m+1} k_1 k_2)^2}
		\\
		\leq  4\int_{\R^3} \diff \qv \f{1}{(1 + q^2 - \frac{2}{m+1}q)^2} \int_{\R^3} \diff \kv_1 \: \one_{\{k_1 \geq \eps\}} \: k_1^{-1} \lf| \eta (\kv_1) \ri|^2 , 
	}
	which yields the r.h.s. of \eqref{pot norm 1}.
	
	\n
	The corresponding lower bound is a bit more involved and requires a deeper inspection of the off-diagonal term  containing $ \eta(\kv_1)^* \eta(\kv_2) $. We introduce in the integrand  the characteristic function of the set $ k_2/k_1 \geq a $, where $ a > 1 $ is a parameter to be chosen later, which is clearly admissible thanks to the positivity of the integrand. Dropping the positive term proportional to $ |\eta(\kv_2)|^2 $, we obtain
	\bml{
		\int_{\R^6} \diff \kv_1 \diff \kv_2 \: \one_{\{k_1 \geq \eps\}} \one_{\{k_2 \geq \eps\}} \: \lf| \pot \eta (\kv_1,\kv_2) \ri|^2 	\\
		\geq \int_{\R^6} \diff \kv_1 \diff \kv_2 \: \one_{\{k_1 \geq \eps\}} \one_{\{k_2 \geq \eps\}}  \one_{\{k_2/k_1 \geq a\}}  \: \frac{ \lf| \eta (\kv_1) \ri|^2 }{(k_1^2 + k_2^2 + \frac{2}{m+1} \kv_1 \cdot \kv_2)^2}  \\
	       - 2    \int_{\R^6} \diff \kv_1 \diff \kv_2 \: \one_{\{k_1 \geq \eps\}} \one_{\{k_2 \geq \eps\}}  \one_{\{k_2/k_1 \geq a\}}  \: \frac{  |\eta(\kv_1)|  |\eta(\kv_2)|}{(k_1^2 + k_2^2 + \frac{2}{m+1} \kv_1 \cdot \kv_2)^2}	
	       \; =: \; (I)  + (II).
	}
	
	\n
	For the diagonal term $(I)$ we have
	\bml{ \label{stI}
	(I) 
	=  \int_{\R^3} \diff \kv_1 \, \one_{\{k_1 \geq \eps\}} \lf| \eta (\kv_1) \ri|^2   \int_{\R^3} \diff \kv_2 \, \one_{\{k_2/k_1 \geq a\}} 
	     \: \frac{1}{(k_1^2 + k_2^2 + \frac{2}{m+1} \kv_1 \cdot \kv_2)^2}\\
	     \geq   \int_{\R^3} \diff \kv_1 \, \one_{\{k_1 \geq \eps\}} \lf| \eta (\kv_1) \ri|^2   \int_{\R^3} \diff \kv_2 \, \one_{\{k_2/k_1 \geq a\}} 
	     \: \frac{1}{(k_1^2 + k_2^2 + \frac{2}{m+1} k_1k_2)^2} \\
	     = 4 \pi \int_a^{\infty} \!\! \diff q\, \f{q^2}{(1 + q^2 + \frac{2}{m+1} q)^2}  \int_{\R^3} \diff \kv_1 \, \one_{\{k_1 \geq \eps\}}   k_1^{-1} \lf| \eta (\kv_1) \ri|^2 \\
	     \; =: \; c_0(a) \int_{\R^3} \diff \kv_1 \, \one_{\{k_1 \geq \eps\}}   k_1^{-1} \lf| \eta (\kv_1) \ri|^2
	}
where $a \,c_0(a) \rightarrow l_0$, with $ 0<l_0 <\infty, \;$ for $\;a \rightarrow \infty$.	For the off-diagonal term we use Schur's test. Denoting by $K$ the integral operator in $L^2(\R^6)$ with (non-symmetric) kernel
\be
K(\kv_1, \kv_2)= 	\one_{\{k_2/k_1 \geq a\}} 
	     \: \frac{k_1^{1/2} k_2^{1/2}}{(k_1^2 + k_2^2 + \frac{2}{m+1} \kv_1 \cdot \kv_2)^2},
	     \ee
	     we have
\bml{
|(II)| \leq 2 \int_{\R^3} \diff \kv_1 \, \one_{\{k_1 \geq \eps\}} k_1^{-1/2} \lf| \eta (\kv_1) \ri|   
\int_{\R^3} \diff \kv_2 \, \one_{\{k_2 \geq \eps\}} k_2^{-1/2} \lf| \eta (\kv_2) \ri| 
 \, K(\kv_1, \kv_2)\\
\leq 2\, \|K\| \int_{\R^3} \diff \kv_1 \, \one_{\{k_1 \geq \eps\}}   k_1^{-1} \lf| \eta (\kv_1) \ri|^2
}
and
\be\label{sch}
\lf\| K \ri\| \leq \bigg[ \sup_{\kv_1 \in \R^3} \int_{\R^3} \diff \kv_2 \: K(\kv_1,\kv_2) \bigg]^{1/2} \bigg[ \sup_{\kv_2 \in \R^3} \int_{\R^3} \diff \kv_1 \: K(\kv_1,\kv_2) \bigg]^{1/2}.
\ee
	Let us estimate the first integral in the r.h.s. of \eqref{sch}
	\bml{
	\int_{\R^3} \diff \kv_2 \: K(\kv_1,\kv_2) \leq \int_{\R^3} \diff \kv_2 \: \one_{\{k_2/k_1 \geq a\}} 
	     \: \frac{k_1^{1/2} k_2^{1/2}}{(k_1^2 + k_2^2 - \frac{2}{m+1} k_1 k_2)^2}\\
	     \leq 4 \pi \int_a^{\infty} \!\! \diff q \, \f{q^{5/2}}{(1 + q^2 - \frac{2}{m+1} q)^2} 
	     =: \; c_1(a)
	}
	where $a^{1/2} c_1(a) \rightarrow l_1$, with $0<l_1<\infty ,\;$ for $\; a \rightarrow \infty$. Analogously, for the second integral we find
	\be
	\int_{\R^3} \diff \kv_1 \: K(\kv_1,\kv_2) \leq 4 \pi \int_0^{1/a} \!\!\! dq \, \f{q^{5/2}}{(1 + q^2 -\frac{2}{m+1} q )^2} \, =: \, c_2(a)
	\ee
	where $a^{7/2} c_2(a) \rightarrow l_2$, with $0<l_2 <\infty,\;$ for $\; a\rightarrow \infty$. Therefore we obtain
	\be\label{stII}
	|(II)| \leq 2\, \sqrt{c_1(a) \, c_2(a)}     \int_{\R^3} \diff \kv_1 \, \one_{\{k_1 \geq \eps\}}   k_1^{-1} \lf| \eta (\kv_1) \ri|^2.      
	\ee
	Taking into account of \eqref{stI}, \eqref{stII} we have
	\be
	(I) + (II) \geq         \f{1}{a} \left( a\, c_0(a) - \f{1}{a} \sqrt{ a^{1/2} c_1(a) \, a^{7/2} c_2(a)} \right)              \int_{\R^3} \diff \kv_1 \, \one_{\{k_1 \geq \eps\}}   k_1^{-1} \lf| \eta (\kv_1) \ri|^2.      
	\ee
	Fixing $a$ sufficiently large, we conclude the proof.
	 	\end{proof}

\vs

\n
The condition \eqref{cond eta 1} does not give any information on the integrability of $ \hat{\eta} $ close to the origin. It is therefore more convenient to exploit an equivalent but slightly different decomposition of functions in $ \dom[\formo] $, where the potential $\pot$ is replaced by $\mathcal G_{\lambda}$, $\lambda >0$, defined as follows

\beq
	\label{potential lambda}
	\lf(\widehat{\pot_{\la} \eta}\ri)\lf(\kv_1,\kv_2\ri) = \frac{\hat{\eta}(\kv_1) - \hat{\eta}(\kv_2)}{k_1^2 + k_2^2 + \tx\frac{2}{m+1} \kv_1 \cdot \kv_2 + \la}.
\eeq 
More precisely, any $u \in \dom[\formo]$ is decomposed as $u=w_{\lambda} + \mathcal G_{\lambda} \eta$, with $ w_{\la} \in \hf(\R^6) $. 
Since $ u \in \ldf(\R^6) $, one has simply to require that 
\be
 \pot_{\la} \eta  \in \ldf(\R^6) 
 \ee
 for any $ \eta $ admissible, i.e., belonging to $ \dom[\formo] $. Next Lemma \ref{lem: pot norm 2} shows that this is equivalent to the requirement
\beq
	\eta \in H^{-1/2}(\R^3),
\eeq 
which has to be satisfied by any charge in \eqref{form domain}.

\begin{lem}
	\label{lem: pot norm 2}
	\mbox{}	\\
	For any $ \la > 0 $, there exist two constants $ 0 < c_1, c_2 < \infty $ independent of $ \la $, such that
	\beq
		\label{pot norm 2}
		c_1 \lf\| \eta \ri\|_{H^{-1/2}_{\la}(\R^3)}^2  \leq \lf\| \pot_{\la} \eta \ri\|_{L^2(\R^6)}^2 \leq c_2  \lf\| \eta \ri\|_{H^{-1/2}_{\la}(\R^3)}^2	\eeq
	where 
	\bdm
				\lf\| \eta \ri\|_{H^{-1/2}_{\la}(\R^3)}^2 = \int_{\R^3} \diff \kv \:  \frac{ \lf| \hat\eta (\kv) \ri|^2}{\sqrt{k^2 + \la}}.	
		\edm

\end{lem}

\begin{proof}
	The r.h.s. of the inequality can be proven as in \cite[Proposition 6.1]{FT} via an estimate analogous to \eqref{pot norm ub}. 
	
	 For the l.h.s. we first bound from below $ \lf\| \pot_{\la} \eta \ri\|_{L^2(\R^6)}^2 $ by cutting the integral domain where $ k_2 \leq a \sqrt{k_1^2 + \la} $, where $ a $ is any positive number to be chosen later large enough. Dropping one of the positive term coming from $ |\eta(\kv_1) - \eta(\kv_2)|^2 $, we get
	\bml{
		\label{pot norm 2 proof 0}
		\lf\| \pot_{\la} \eta \ri\|_{L^2(\R^6)}^2 \geq   \int_{\R^6} \diff \kv_1 \diff \kv_2 \: \one_{\lf\{k_2 \geq  a\sqrt{k_1^2 + \la}\ri\}} \: \frac{ \lf| \hat\eta (\kv_1) \ri|^2 }{(k_1^2 + k_2^2 + \frac{2}{m+1} \kv_1 \cdot \kv_2 +\lambda )^2}	\\
		-2\int_{\R^6} \diff \kv_1 \diff \kv_2 \: \one_{\lf\{k_2 \geq  a\sqrt{k_1^2 + \la}\ri\}} \: \frac{ |\hat\eta(\kv_1)||\hat\eta(\kv_2)|}{(k_1^2 + k_2^2 + \frac{2}{m+1} \kv_1 \cdot \kv_2 +\lambda )^2} \, =: \, (III) + (IV).
	}
	The positive term $(III)$ can be  estimated from below  as
	\bml{		\label{pot norm 2 proof 1}
		(III) \geq \int_{\R^3} \diff \kv_1  \: \lf|\hat\eta (\kv_1) \ri|^2
		\int_{\R^3} \diff \kv_2\, 
		\one_{\lf\{k_2 \geq  a\sqrt{k_1^2 + \la}\ri\}} \: \frac{ 1 }{(k_1^2 + k_2^2 + \frac{2}{m+1} k_1 k_2 +\lambda)^2}\\
		= 4\pi \int_{\R^3} \diff \kv_1  \: \f{ \lf|\hat\eta (\kv_1) \ri|^2}{\sqrt{k_1^2 +\lambda}} \int_a^{\infty}\!\!\! \diff q \, \f{q^2}{(1+q^2+\frac{2}{m+1} \frac{k_1}{\sqrt{k_1^2 +\lambda}} q)^2} \, \geq \, c_0(a) \, \lf\| \eta \ri\|_{H^{-1/2}_{\la}(\R^3)}^2
		}
where $c_0(a)$ has been defined in \eqref{stI}.	For the negative term $(IV)$ we write
\bml{
|(IV)| = 2 \int_{\R^3} \diff \kv_1  \: \f{ \lf|\hat\eta (\kv_1) \ri|^2}{ (k_1^2 +\lambda)^{1/4}}
\int_{\R^3} \diff \kv_2  \: \f{ \lf|\hat\eta (\kv_2) \ri|^2}{ (k_2^2 +\lambda)^{1/4}} \, K_{\lambda}(\kv_1, \kv_2) \leq 2  \, \|K_{\lambda}\| \,  \lf\| \eta \ri\|_{H^{-1/2}_{\la}(\R^3)}^2
}
where $K_{\lambda}$ is the integral operator in $L^2(\R^3)$ with kernel
\be
K_{\lambda}(\kv_1, \kv_2) = \one_{\lf\{k_2 \geq  a\sqrt{k_1^2 + \la}\ri\}} \: \frac{ (k_1^2 + \lambda)^{1/4} (k_2^2 + \lambda)^{1/4} }{(k_1^2 + k_2^2 + \frac{2}{m+1} \kv_1 \cdot \kv_2 +\lambda )^2}.
\ee	
We use again Schur's test to estimate $\|K_{\lambda}\|$.
\bml{
\int_{\R^3} \diff \kv_2 \, K_{\lambda} (\kv_1, \kv_2) \leq 4 \pi (k_1^2 + \lambda)^{1/4} \int_{a\sqrt{k_1^2 + \lambda} }^{\infty} \!\!\! \diff k_2 \, \f{k_2^2 \,  (k_2^2 + \lambda)^{1/4}}{(k_1^2 + k_2^2 -\frac{2}{m+1} k_1 k_2 + \lambda)^2} \\= 4 \pi \int_a^{\infty} \!\!\!\diff q\, \f{q^2 \left( q^2 +\lambda (k_1^2 + \lambda)^{-1} \right)^{1/4} }{(1+q^2 -\frac{2}{m+1} \frac{k_1}{\sqrt{k_1^2 +\lambda}} q )^2} \leq 
4 \pi \int_a^{\infty} \!\!\!\diff q\, \f{q^{5/2} \left( 1+q^{-2} \right)^{1/4} }{(1+q^2 -\frac{2}{m+1}  q )^2} \, =: \, \tilde{c}_1 (a),
}	
where  $a^{1/2} \tilde{c}_1(a) \rightarrow \tilde{l}_1$, with $0<\tilde{l}_1<\infty ,\;$ for $\; a \rightarrow \infty$.   Moreover           
\bml{
\int_{\R^3} \diff \kv_1 \, K_{\lambda} (\kv_1, \kv_2) \leq 4 \pi (k_2^2 + \lambda)^{1/4} \int_0^{\frac{\sqrt{k_2^2 -\lambda a^2}}{a}} \!\!\! \diff k_1 \, \f{k_1^2 \, (k_1^2 + \lambda)^{1/4}}{(k_1^2 + k_2^2 - \frac{2}{m+1} k_1 k_2 + \lambda)^2} \\
= 4 \pi \int_0^{q(a)} \!\!\! \diff q \, \f{ q^2 (q^2 + \lambda (k_2^2 +\lambda)^{-1} )^{1/4}}{(1 + q^2 - \frac{2}{m+1} \frac{k_2}{\sqrt{k_2^2 + \lambda}} q)^2},
}
where $q(a) = \f{\sqrt{k_2^2 - \lambda a^2}}{a \sqrt{k_2^2 + \lambda}} \leq a^{-1}$. Therefore
\be
\int_{\R^3} \diff \kv_1 \, K_{\lambda} (\kv_1, \kv_2) \leq
4 \pi \int_0^{1/a} \!\!\! \diff q \, \f{ q^2 (q^2 + 1 )^{1/4}}{(1 + q^2 - \frac{2}{m+1}  q)^2} \, =: \, \tilde{c}_2(a),
\ee
where  $a^{3} \tilde{c}_2(a) \rightarrow \tilde{l}_2$, with $0<\tilde{l}_2<\infty ,\;$ for $\; a \rightarrow \infty$. Hence we obtain the estimate $\|K_{\lambda}\| \leq \sqrt{ \tilde{c}_1(a) \, \tilde{c}_2(a)}$. We can now conclude the proof proceeding along the same line of Lemma \ref{lem: pot norm 1}. 
\end{proof}

\vs
\n
According to the decomposition $u=w_{\lambda} + \mathcal G_{\lambda} \eta$, with $ w_{\la} \in \hf(\R^6) $, for any $u \in \dom[\formol]$, and the charge decomposition $\eta= \xi_< + \Xi_>$, with $\xi_< \in H^{1/2}(\R^3)$ and $\Xi_> \in H^{-1/2}(\R^3)$,  introduced in  \eqref{fine1} and \eqref{fine2}, in the following proposition we derive an equivalent expression for our quadratic form.

	\begin{pro}[Alternative expression of $ \formol $]
		\label{pro: alternative form}
		\mbox{}	\\
		For any $ \beta$ and $ \la > 0 $, the quadratic form \eqref{formol domain}, \eqref{formol action} can be equivalently rewritten as 
				\begin{eqnarray}
 			\dom[\formol] &=& \lf\{ \lf. u \in \ldf(\R^6) \: \ri| \: u = w_{\la} + \pot_{\la} \eta, w_{\la} \in \hf(\R^6), \eta \in H^{-1/2}(\R^3),  \eta = \xi_< + \Xi_>, \ri. \nonumber	\\
 			&& \lf. \xi_< \in H^{1/2}(\R^3) \ri\},	\label{alternative form domain} \\
 			\formol[u] & = & \uform_{\la}[w_{\la}] - \la \lf\| u \ri\|_{\ldf(\R^6)}^2 + 2 \qformla[\xi_<] +2 \qformo[\Xi_<] - 4 \Re \qformo[\xi_<, \Xi_< ] + 2 \la \Re \Big( \pot \Xi_>, \pot_{\la} \xi_< \Big)_{L^2(\R^6)}  \nonumber	\\
 			& &+ \la \Big( \pot \Xi_>, \pot_{\la} \Xi_> \Big)_{L^2(\R^6)} + \beta |q|^2 ,\label{alternative form action}
		\end{eqnarray}
		where 
		\begin{eqnarray}
		F_{\la}[w] &=& \lf(w, (\hamfree + \la) w \ri) , \\
	\label{qformla}
	\qformla[\xi] &=& 2\pi^2 \frac{\sqrt{m(m+2)}}{m+1} \int_{\R^3} \diff \pv \: \sqrt{p^2 + \la} \: |\hat{\xi}(\pv)|^2 + \int_{\R^6} \diff \pv \diff \qv  \: \frac{\hat{\xi}^*(\pv) \hat{\xi}(\qv)}{p^2 + q^2 + \frac{2}{m+1} \pv \cdot \qv + \la}.
\end{eqnarray}
and $\Phi_0[\cdot, \cdot]$ denotes the bilinear form associated to $\Phi_0[\cdot]$.
			\end{pro}
	
	\begin{proof}
		Uniqueness of the decomposition $ u = w_{\la} + \pot_{\la} \eta $ follows from the fact that $ w_{\la} \in \hf(\R^6) $ while $  \pot_{\la} \eta \notin \hf(\R^6) $, for any $ \eta $. Moreover the domain is independent of $ \la $ as a consequence of the regularity of $ (\pot_{\la} - \pot_{\mu}) \eta $ for any $ \la, \mu > 0 $. Indeed, the resolvent identity yields
		\bdm
			(\pot_{\la} - \pot_{\mu}) \eta = \lf( \mu - \la \ri) \lf( \hamfree + \mu \ri)^{-1} \pot_{\la} \eta,
		\edm
		which is clearly in the domain of $ \hamfree $, i.e., $ \hdf(\R^6) $, any time $ \pot_{\la} \eta \in \ldf(\R^6) $. It remains then to show that the equivalence holds true also for $ \mu = 0 $. In this case it suffices to set
		\beq
			w : = w_{\la} + \lf( \pot_{\la} - \pot \ri) \eta,
		\eeq
		which belongs to $ \dot{\hf}(\R^6) $, if $ w_{\la} \in \hf(\R^6) $ and $ \pot \eta \in L^2_{\mathrm{loc}}(\R^6) $, and this follows from the condition $ \xi \in \dot{H}^{1/2}(\R^3) $ and a direct inspection of the other term $ \pot \Xi^- $.
	
	\n
		The expression of the quadratic form is a direct consequence of the above definition and therefore we give only a sketch of the computation. Writing $w=w_{\lambda} + (\pot_{\la}-\pot) \eta$ and $ \eta = \xi_< + \Xi_>$, we find
		\bml{
		\formol[u]= \left(w_{\la},(\hamfree +\la)w_{\la} \right)_{L^2(\R^6)} -\la \|u\|_{\ldf(\R^6)}^2 + \la \Big( \pot (\xi_< + \Xi_>), \pot_{\la} (\xi_< + \Xi_>) \Big)_{L^2(\R^6)} \\
		+ 2 \Phi_0[\xi_< - \Xi_>] + \beta |q|^2\\
		= F_{\la}[w_{\la}]  -\la \|u\|_{\ldf(\R^6)}^2 + 2 \qformo[\Xi_<] - 4 \Re \qformo[\xi_<, \Xi_< ] + 2 \la \Re \Big( \pot \Xi_>, \pot_{\la} \xi_< \Big)_{L^2(\R^6)}  	\\
 			+ \la \Big( \pot \Xi_>, \pot_{\la} \Xi_> \Big)_{L^2(\R^6)} + \beta |q|^2 + \la \Big( \pot \xi_< , \pot_{\la} \xi_<  \Big)_{L^2(\R^6)}+ 2 \Phi_0[\xi_< ]
		}
		By an explicit computation  we obtain
		
		\be
		\la \Big( \pot \xi_< , \pot_{\la} \xi_<  \Big)_{L^2(\R^6)}+ 2 \Phi_0[\xi_< ] = 2 \Phi_{\la}[\xi_<]
		\ee
		and then \eqref{alternative form action} is proved. Notice that all the terms of the quadratic form in \eqref{alternative form action} are separately finite thanks to the properties of $ \xi_< $, $ \Xi_< $ and $ \Xi_> $.
	\end{proof}


\begin{thebibliography}{99}

\bibitem[AGH-KH]{al} \textsc{Albeverio S., Gesztesy F., Hoegh-Krohn R., Holden H.}, {\em Solvable Models in Quantum Mechanics}, Springer-Verlag, New-York, 1988.

\bibitem[AS]{simon} \textsc{Alonso A., Simon B.}, The Birman-Krein-Vishik theory of self-adjoint extensions of semi-bounded operators, {\em J. Operator Theory}, {\bf 4} (1980), 251--270.

\bibitem[B]{bi} \textsc{Birman M.S.}, On the self-adjoint extensions of positive definite operators (in Russian), {\em Math. Sb.}, {\bf 38} (1956), 431--450; English translation available on preprint {\it SISSA 08/2015/MATE} (url: http://urania.sissa.it/xmlui/handle/1963/34443).


\bibitem[BH]{bh} \textsc{Braaten E., Hammer H.W.}, Universality in few-body systems with large scattering length, {\em Phys. Rep.}, {\bf 428}  (2006), 259--390.

\bibitem[CMP]{cmp} \textsc{Castin Y., Mora C., Pricoupenko L.}, Four-Body Efimov Effect for Three Fermions and a Lighter Particle, {\em Phys. Rev. Lett.} {\bf 105} (2010), 223201.

\bibitem[CT]{ct}		\textsc{Castin Y., Tignone E.}, Trimers in the resonant $(2\!+\!1)-$fermion problem on a narrow Feshbach resonance: Crossover from Efimovian to hydrogenoid spectrum, {\it Phys. Rev. A} {\bf 84} (2011), 062704.

\bibitem[CW]{cw} \textsc{Castin Y., Werner F.}, The Unitary Gas and its Symmetry Properties. In {\em Lect. Notes Phys.} {\bf 836} (2011) 127-189. 



\bibitem[CDFMT]{CDFMT} 	\textsc{Correggi M., Dell'Antonio G., Finco D., Michelangeli A., Teta A.}, Stability for a System of $ N $ Fermions Plus a Different Particle with Zero-Range Interactions, {\it Rev. Math. Phys.} {\bf 24} (2012), 1250017.

\bibitem[CFT]{cft} \textsc{Correggi M., Finco D., Teta A.},  Energy lower bound for the unitary N + 1 fermionic model, {\it Europhys. Lett.} {\bf 111} (2015), 10003.










\bibitem[DFT]{DFT}	\textsc{Dell'Antonio G.,  Figari R.,  Teta A.}, Hamiltonians for Systems of $ N $ Particles Interacting through Point Interactions, {\it Ann. Inst. H. Poincar\'{e} Phys. Th\'{e}or.} {\bf 60} (1994), 253--290.

\bibitem[E]{E} \textsc{V. Efimov}, Energy levels of three resonantly interacting particles, {\it Nucl. Phys. A} {\bf 210} (1973), 157.

\bibitem[FM]{fm} \textsc{Faddeev L., Minlos R.A.}, On the point interaction for a three-particle system in Quantum Mechanics, {\em Soviet Phys. Dokl.}, {\bf 6} (1962), 1072--1074.


\bibitem[FT]{FT}		\textsc{Finco D.,  Teta A.}, Quadratic Forms for the Fermionic Unitary Gas Model, {\em Rep. Math. Phys.} {\bf 69} (2012), 131--159.





\bibitem[KM]{km}  \textsc{Kartavtsev, O. I.,  Malykh, A. V.} Recent advances in description of few two- component fermions. 
{\it Physics of atomic nuclei}  {\bf 77},  (2014) 430-437. 

\bibitem[MS]{ms} \textsc{Michelangeli A., Schmidbauer C.}, Binding properties of the (2+1)-fermion system with zero-range interspecies interaction. {\em Phys. Rev. A} {\bf 87} (2013), 053601.

\bibitem[M1]{m1} \textsc{Minlos R.A.}, On the point interaction  of  three particles, {\em Lect. Notes in Physics} {\bf 324}, Springer, 1989.


\bibitem[M2]{m2}\textsc{Minlos R.A.}, On Pointlike Interaction between Three Particles:
Two Fermions and Another Particle, {\it ISRN Mathematical Physics} (2012), 230245.

\bibitem[M3]{m3}  \textsc{Minlos R.A.}, On point-like interaction between $n$ fermions and another particle,  {\em Moscow Math. Journal}, {\bf 11}  (2011),  113--127.

\bibitem[M4]{m4} \textsc{Minlos R.A.}, A system of three quantum particles with point-like interactions. {\em Russian Math. Surveys} {\bf 69} (2014), 539-564.








\bibitem[TC]{tc} \textsc{Trefzger C., Castin Y.} Self-energy of an impurity in an ideal Fermi gas to second order in the interaction strength. {\em Phys. Rev. A} {\bf 90} (2014), 033619.

\bibitem[W]{wth} \textsc{Werner F.}, Ph.D. Thesis, \'{E}cole Normale Sup\'{e}rieure, 2008.

\bibitem[WC1]{wc1} \textsc{Werner F., Castin Y.}, Unitary gas in an isotropic harmonic trap: symmetry properties and applications, {\em Phys. Rev. A} {\bf 74}  (2006), 053604.

\bibitem[WC2]{wc2} \textsc{Werner F., Castin Y.}, Unitary Quantum Three-Body Problem in a Harmonic Trap. {\em Phys. Rev. Lett.} {\bf 97} (2006), 150401.

 
\end{thebibliography}
\end{document}